\documentclass{article}
\usepackage{amsmath}

\usepackage{natbib}
\usepackage{amsthm}
\theoremstyle{plain} 
\newtheorem{theorem}{Theorem}[section]      
\newtheorem{proposition}{Proposition}[section] 
\newtheorem{assumption}{Assumption}[section]
\usepackage{graphicx}

\usepackage{booktabs}
\usepackage{tabularx}
\usepackage{newtxtext}
\usepackage[subscriptcorrection]{newtxmath}

\newcommand{\dd}{\mathop{}\!\mathrm{d}} 

\usepackage[plain,noend]{algorithm2e}

\makeatletter
\renewcommand{\algocf@captiontext}[2]{#1\algocf@typo. \AlCapFnt{}#2} 
\def\@algocf@capt@plain{top}
\renewcommand{\algocf@makecaption}[2]{%
  \addtolength{\hsize}{\algomargin}%
  \sbox\@tempboxa{\algocf@captiontext{#1}{#2}}%
  \ifdim\wd\@tempboxa >\hsize
    \hskip .5\algomargin%
    \parbox[t]{\hsize}{\algocf@captiontext{#1}{#2}}
  \else%
    \global\@minipagefalse%
    \hbox to\hsize{\box\@tempboxa}
  \fi%
  \addtolength{\hsize}{-\algomargin}%
}
\makeatother

\begin{document}



\markboth{S. I WATSON and J. FRASER-GOVIL}{Causal policy effects for spatial exposure fields}

\title{Identification of causal policy effects for spatial exposure fields}

\author{S. I. WATSON and J. FRASER-GOVIL}

\maketitle

\begin{abstract}
Many interventions act on an entire spatial exposure field. We study when their induced laws satisfy policy positivity and what its failure implies for causal evaluation through weighting and structural modelling. We show that, for post-transformation fields that are Gaussian, the only admissible translation that satisfy positivity lie on the Cameron-Martin space of the observed exposure field. Under fine-scale variation and regularity assumptions, these are the only admissible non-decreasing pointwise policies on that scale. Common rules, including binding caps and proportional reductions on the Gaussian scale can therefore violate positivity. The Cameron-Martin norm controls both weight variability and the amplification of disagreement between candidate outcome models into disagreement between policy estimates. Without positivity, bounded continuous response models can make arbitrarily similar observational predictions while their policy estimates remain separated. We also give a policy-specific criterion for structural identification conditional on a realised exposure state. Finally, when a policy violating positivity admits weights in every finite representation, their second moments diverge under nested refinement that recovers the full field. These results motivate assessing identification and sensitivity for the specified policy effect across plausible structural models and under refinement of finite representations.
\end{abstract}


\section{Introduction}
\citet{DefraTargets2023} set two legally binding targets for PM2.5 pollution by 2040: an annual mean concentration of $10\,\mu\mathrm{g\,m^{-3}}$, and a $35\%$ reduction in population exposure against a 2016-18 baseline. The two are presented as complementary, the first acting where concentrations are highest and the second on exposure everywhere. Evaluating the health effect of either means comparing outcomes under the pollution surface the target would produce with outcomes under the surface observed. We show that the two questions are not equally answerable. The effect of the proportional reduction can be identified by reweighting the observed exposure process. The effect of the cap can though can only be evaluated through restrictions on a response model. However, for caps and similar policies, different response models can make arbitrarily similar observational predictions while their policy estimates remain separated
 
Many interventions of policy interest act in this way on an entire surface or spatial configuration rather than on a scalar treatment, such as pollution abatement, climate effects, or the siting of services. Representing the exposure as a field, causal evaluation requires comparing probability laws on an infinite-dimensional treatment space. Effects of stochastic policies are identified by reweighting the observational treatment distribution when the target law $\pi(\cdot\mid C)$ is absolutely continuous with respect to the observational law $g_{0}(\cdot\mid C)$, so that a Radon-Nikodym policy weight exists. This is the measure-theoretic form of the requirement that treatment levels relevant to the estimand occur within covariate strata, and it underlies stochastic and modified treatment policies for continuous exposures \citep{Westreich2010,Petersen2012,Munoz2012,Haneuse2013}. Such policies are attractive precisely because they can be chosen to stay close to the observed treatment distribution, and so to alleviate conventional positivity problems.
 
For an intervention that can be described as a field absolute continuity is far more restrictive than finite-dimensional overlap suggests. Two laws on an 
infinite-dimensional space can be mutually singular even when every finite-dimensional projection has an ordinary density. The distinction is classical for Gaussian measures \citep{Cameron1944,Feldman1958}, and has been used to show the restrictiveness of overlap when the \emph{covariates} are functional \citep{Ghosh2019}. Where the exposure itself is functional, a growing literature works through finite-dimensional representations \citep{Zhang2021,Tan2025,Wang2024,Ciardulli2026,Kennedy2019}, obtaining weighting representations under uniform moment bounds on the density ratios \citep{Jiang2026}, or constructing stochastic policies designed to avoid restrictive positivity requirements \citep{Barnard2026}. This leaves open a question that precedes estimation: when does a scientifically specified transformation of a continuum treatment induce a law absolutely continuous with respect to the observational law?

For fields that are Gaussian after a prespecified transformation, we characterise admissible translations through the Cameron-Martin space. Under conditions on the field's fine-scale variation and the regularity of the policy, these are the only admissible non-decreasing pointwise transformations on that scale. The scale also plays an important role as translation on the logarithmic scale acts multiplicatively on the original exposure surface. The Cameron-Martin norm also connects positivity to structural modelling. It bounds how much two outcome-response models can disagree about a policy effect when their predictions agree under the observational exposure law. A large norm permits greater amplification of this disagreement, even when the policy is admissible. Without positivity, uniform control does not hold and bounded continuous response models can make arbitrarily similar observational predictions while implying separated policy effects. 

Failure of positivity nevertheless permits identification under suitable structural restrictions. We give a policy-specific criterion conditional on a realised exposure state, clarifying which contrasts a stated response model identifies. We also examine exposure reconstruction where applying a cap to a conditional mean surface can understate the expected exposure removed. Finally, finite representations can admit policy weights even when the continuum policy has none. If weights exist at every stage of a nested refinement recovering the full field, their second moments must diverge. The same divergence describes the loss of uniform control over disagreement between model-based policy estimates. 

\section{Setup}
\subsection{Statistical Model and Potential Outcomes}
Let $\mathcal D\subset\mathbb R^{d}$ be a contiguous spatial or spatio-temporal domain containing $N$ outcome units at locations $s_{1},\ldots,s_{N}\in\mathcal D$, and let $D$ denote the largest scientifically relevant separation between an outcome unit and a point of $\mathcal D$. We observe baseline information $C$, which contains measured prognostic and spatial-assignment variables as well as fixed geographic information. 

Let $\rho$ be a prespecified reference measure on $\mathcal D$, such as area, network length or population support, which fixes the units in which spatial exposure is expressed. Let $\mathscr A$ be a set of finite non-negative measures on $\mathcal D$, equipped with the $\sigma$-algebra generated by the evaluations $a\mapsto a(B)$ for Borel $B\subseteq\mathcal D$. A state $a\in\mathscr A$ records how much intervention is present where. We assume that $\mathcal D$ is Borel, that $\rho$ is $\sigma$-finite, and that $\mathscr A$, equipped with its evaluation $\sigma$-algebra, and the state space of $C$ are standard Borel. All transformations depending on $C$ are jointly measurable.

A \emph{spatial exposure field}, such as a pollutant surface, is a state with $a\ll\rho$. We write $\dot a=\dd a/\dd\rho$ for its concentration surface, so that $a(B)=\int_{B}\dot a\,\dd\rho$.  For every observational or target treatment-state law $q(\cdot\mid C)$ considered below, assume $\int_{\mathscr A} a(\mathcal D)\,q(\dd a\mid C)<\infty$ almost surely. All spatial kernels appearing below are bounded and measurable.

Let $d(s,t)\in[0,D]$ be a prespecified spatial, travel-time, network or other scientifically relevant distance, and let $\chi_i(t)\in\{0,1\}$ indicate which exposure locations are included in the spatial exposure summary for unit $i$. For a bounded measurable kernel $\phi$, define
\begin{equation}
    S_i[\phi](a)
    =\int_{\mathcal D}\chi_i(t)\,\phi\{d(s_i,t)\}\,a(\dd t),
    \label{eq:spatial_exposure_mapping}
\end{equation}
and write
$S[\phi](a)=\{S_1[\phi](a),\ldots,S_N[\phi](a)\}^{\top}$. For a field, $a(\dd t)=\dot a(t)\rho(\dd t)$, so this is a spatially weighted integral of concentration.

We also define local exposure by
\begin{equation}
    x_i(a)=\frac{a(\mathcal W_i)}{\rho(\mathcal W_i)},
    \label{eq:local_exposure}
\end{equation}
where $\mathcal W_i$ is a prespecified local window with $0<\rho(\mathcal W_i)<\infty$. This gives a local concentration average for fields, and a binary treatment indicator when $a(\mathcal W_i)\in\{0,1\}$ and $\rho(\mathcal W_i)=1$. Window averages are unchanged by altering a density on a set of $\rho$-measure zero, so they remain well defined when the density is specified only almost everywhere. The local and spatial exposure summaries may overlap. If the spatial term is intended to include only exposure outside the local window, take $\chi_i(t)=0$ on $\mathcal W_i$. Otherwise, changing exposure within the window can change both summaries. Setting $x_i\equiv0$ omits the separate local term.

We treat the distance metric and eligibility mapping as prespecified parts of the causal question. This differs from work that relaxes the requirement that an exposure mapping capture the full interference structure \citep{Savje2024,Leung2022}; uncertainty or misspecification of the mapping is not addressed by the identification results below, although \S\ref{sec:model_free_displacement} bears on the comparison of candidate mappings directly.

Let $Y_{i}(x,a)$ denote the potential outcome under joint intervention on unit $i$'s own treatment status $x$ and the spatial treatment state $a$. A state induces the own status $x_{i}(a)$ and hence the state-specific potential outcome $Y_{i}(a)=Y_{i}\{x_{i}(a),a\}$. For each $i$, assume that there exist a random element $U_i$ in a standard Borel space and a jointly measurable function $F_i$ such that, outside a single null set, $Y_i(a)=F_i(a,C,U_i)$ for every $a\in\mathscr A$. Writing $Q_i(\cdot\mid C)$ for a regular conditional law of
$U_i$ given $C$, conditional potential-outcome means and the response model below use the versions
\begin{equation*}
    \mu_i(a,C) =\int F_i(a,C,u)\,Q_i(\dd u\mid C)
\end{equation*}
Let $g$ be a known link function. We assume the linear spatial response model
\begin{equation}
    \mathbb E\{Y_{i}(a)\mid C\}
    =\eta_{0i}(C)+\tau x_{i}(a)+S_{i}[\phi](a),
    \label{eq:potential_outcome_model_point_process}
\end{equation}
where $\eta_{0i}(C)$ is the baseline linear predictor, $\tau$ is the own-treatment effect and $\phi$ is the true spatial spillover kernel. 

Estimation replaces $\eta_{0i}(C)$ by a baseline design $H_{i}(C)$ with coefficient $\gamma$, and restricts $\phi$ to a prespecified class $\mathcal K$ of kernels on $[0,D]$, giving the working linear predictor
\begin{equation}
    \mu_{i}(a,C)
    =H_{i}(C)^{\top}\gamma+\tau x_{i}(a)+S_{i}[\phi](a),
    \qquad \phi\in\mathcal K .
    \label{eq:working_predictor}
\end{equation}

\subsection{Target policies and the estimand}
\label{sec:estimand}
Estimands are defined conditionally on the realised baseline and geographic information $C$ on the finite domain $\mathcal D$, and are therefore $C$-measurable random variables. First, we write $A\mid C\sim g_{0}(\dd a\mid C)$ for the observational law of the realised spatial treatment state, i.e. the mechanism by which the intervention actually arose in the observed system. A target spatial policy is a probability kernel $\pi(\cdot\mid C)$ on $\mathscr A$. The estimand is then, in effect, a difference in expected potential outcomes between two policies. 

A policy may be specified directly, or as the pushforward of a reference law. Throughout, we take this reference law to be $g_0(\cdot\mid C)$ although it may be any policy on $\mathscr A$. A pushforward of this law is one that, under a measurable transformation $\Phi(\cdot,C):\mathscr A\to\mathscr A$, draws a state from $g_0$ and applies $\Phi$ to it, written $ \pi=\Phi(\cdot,C)_{\#}g_0$ that is $\pi(\mathcal L\mid C) =g_0\{a:\Phi(a,C)\in\mathcal L\}$ for measurable $\mathcal L\subseteq\mathscr A$. Most policies of scientific or regulatory interest are \emph{pointwise}: what happens at a location depends only on the concentration there. Such a policy acts through a rule $\varphi$ on the concentration surface,
\begin{equation}
    \dot a \;\longmapsto\; \varphi\{\cdot\,,\dot a(\cdot)\},
    \label{eq:pointwise_policy}
\end{equation}
with proportional reduction, capping and targeted remediation the cases $\varphi(t,u)=ru$, $\min\{u,c(t)\}$ and $\{u-\delta(t)\}_{+}$.

The mean outcome, $\Psi(\pi)$, and mean exposure, $\mu_{\pi}(B\mid C)$  under policy $\pi$ are:
\begin{equation}
    \Psi(\pi)=\frac{1}{N}\sum_{i=1}^{N}
    \int_{\mathscr A}\mu_{i}(a,C)\,\pi(\dd a\mid C), \qquad \mu_{\pi}(B\mid C)=\int_{\mathscr A}a(B)\,\pi(\dd a\mid C).
    \label{eq:policy_mean}
\end{equation} 
For a target policy and our reference law writing $\nu(\dd a\mid C)=\pi(\dd a\mid C)-g_{0}(\dd a\mid C)$, the estimand is the contrast $\Theta_{\nu}=\Psi(\pi)-\Psi(g_0)$. Two summaries of $\nu$ carry the policy information. The induced direct-treatment contrast is
\begin{equation}
    c_{\nu}=\frac{1}{N}\sum_{i=1}^{N}\int_{\mathscr A}x_{i}(a)\,
    \nu(\dd a\mid C),
    \label{eq:policy_direct_contrast}
\end{equation}
For the spatial component, the policy matters through how much exposure mass it adds or removes at each distance from outcome units. Writing $\mu_{\nu}=\mu_{\pi}-\mu_{g_{0}}$ for the signed policy-induced exposure measure, the policy exposure profile by distance is
\begin{equation}
    w_{\nu}(B)=\frac{1}{N}\sum_{i=1}^{N}\int_{\mathcal D}
    \chi_{i}(t)\mathbf 1\{d(s_{i},t)\in B\}\,\mu_{\nu}(\dd t\mid C),
    \qquad B\subseteq[0,D]\ \text{Borel}.
    \label{eq:policy_exposure_distance_measure}
\end{equation}
Because $\nu$ has total mass zero the baseline predictor cancels, and under the identity link $\Theta_{\nu}=c_{\nu}\tau+\langle\phi,w_{\nu}\rangle$. Thus the policy enters the structural estimand only through its direct-treatment contrast and its exposure-by-distance signature.

\subsection{Identification conditions}
\label{sec:identification_conditions}
 
Write $\overline Y=N^{-1}\sum_{i=1}^{N}Y_{i}$ for the domain mean outcome. For identification of $\Theta_\nu$ by weighting, we assume:
\begin{itemize}
    \item \textbf{Consistency.} $Y_{i}=Y_{i}(A)$.
    \item \textbf{Conditional exchangeability.}
    $U_i\perp\!\!\!\perp A\mid C$ for every $i$.
    \item \textbf{Policy positivity.}
    $\pi(\cdot\mid C)\ll g_{0}(\cdot\mid C)$ for almost every $C$.
\end{itemize}
 
Conditional exchangeability is the assumption that adjustment for $C$ removes confounding of intervention placement, which is the spatial confounding problem addressed elsewhere \citep{Hodges2010, Reich2006}. For an exposure field, consistency additionally requires that $a$ capture every outcome-relevant aspect of the intervention, so that two interventions producing the same field produce the same potential outcome, which is the spatial form of the requirement that treatment versions be causally irrelevant or explicitly represented. Policy positivity is the most demanding of the three and the one we explore below. 

\subsection{Two routes to identification}
\label{sec:two_routes}
We consider two routes to identification: nonparametrically by inverse probability weighting, and structurally through the linear outcome model. For the weighting route, under policy positivity the Radon-Nikodym derivative
\begin{equation}
    \omega(a,C)
    =\frac{\dd\pi(\cdot\mid C)}{\dd g_{0}(\cdot\mid C)}(a)
    \label{eq:policy_rn_weight}
\end{equation}
exists and is, for each fixed $C$, determined up to a $g_{0}(\cdot\mid C)$-null set. We choose a version $\omega(a,C)$ jointly measurable in $(a,C)$. Under consistency, conditional exchangeability, policy positivity and the integrability condition $\mathbb E_{g_{0}}[\{1+\omega(A,C)\}\sum_{i=1}^{N}|Y_{i}|]<\infty$, a change of measure gives
\begin{equation}
    \Psi(\pi)=\mathbb E_{g_{0}}\{\omega(A,C)\,\overline Y\mid C\},
    \qquad
    \Theta_{\nu}=\mathbb E_{g_{0}}[\{\omega(A,C)-1\}\,\overline Y\mid C].
    \label{eq:policy_weighting_identity}
\end{equation}
Proofs of \eqref{eq:policy_weighting_identity} and other formal statements below are given in the Appendix.

The identity requires no outcome model, so wherever policy positivity holds, $\Theta_{\nu}$ is a nonparametrically identified causal contrast rather than an artefact of the additivity in \eqref{eq:potential_outcome_model_point_process}. Related weighting constructions have been developed for stochastic interventions on repeatedly observed spatial point patterns \citep{Papadogeorgou2022} and, more recently, for function-valued treatments through finite-dimensional score representations \citep{Zhang2021,Jiang2026,Ciardulli2026}. Our concern is instead whether a policy specified directly as a transformation of the continuum treatment state induces a law absolutely continuous with respect to the observational law in the first place.

Now suppose positivity fails, so that no policy weight exists. The contrast may still be recoverable if we are prepared to restrict the response model using the fact that under an identity link everything is linear. The estimand $\Theta_{\nu}=c_{\nu}\tau+\langle\phi,w_{\nu}\rangle$ is a linear functional of the unknown $(\tau,\phi)$. A single realisation delivers the $N$ conditional means $\mathbb E(Y_{i}\mid A,C)$, which are $N$ further linear functionals of the same unknowns. Identification is then a question of whether the estimand's functional lie in the span of the observable ones. 
 
Let $\sigma_i(a)$ record the exposure included by $\chi_i$, grouped by distance from unit $i$. Specifically, for a Borel
set $B\subseteq[0,D]$, define
\begin{equation*}
    \sigma_i(a)(B)
    =
    \int_{\mathcal D}
        \chi_i(t)\,\mathbf 1\{d(s_i,t)\in B\}\,a(\dd t).
\end{equation*}
Thus $\sigma_i(a)([r_1,r_2])$ is the included exposure mass between distances $r_1$ and $r_2$ from unit $i$, and $S_i[\phi](a)=\langle\phi,\sigma_i(a)\rangle$. Then, in addition to consistency and conditional exchangeability, assume:
\begin{itemize}
    \item[(S1)] identity link, with $\phi\in\mathcal K$ a linear subspace of     bounded measurable functions on $[0,D]$;
    \item[(S2)] $\eta_{0i}(C)=H_{i}(C)^{\top}\gamma+u_{i}$, where     $H_{i}(C)\in\mathbb R^{q}$ is a known vector of baseline covariates with unknown coefficient $\gamma$, and $u\in\mathcal A$ for a known linear subspace $\mathcal A\subseteq\mathbb R^{N}$. Writing $H$ for the matrix with rows $H_{i}(C)^{\top}$, set $\mathcal B=\operatorname{col}(H)+\mathcal A$;
    \item[(S3)] $\theta=(\gamma,u,\tau,\phi)$ unrestricted in  $\Xi=\mathbb R^{q}\times\mathcal A\times\mathbb R\times\mathcal K$;
    \item[(S4)] a single realisation: the observational law determines     $\{\mathbb E(Y_{i}\mid A,C)\}_{i=1}^{N}$ and nothing further about $\theta$.
\end{itemize}
The subspace $\mathcal A$ encodes the latent spatial structure and $\mathcal B$ is the resulting nuisance space. The richer $\mathcal A$, the more exposure-response variation is indistinguishable from latent spatial variation. Call $\Theta_{\nu}$ identified if it takes a common value at every $\theta\in\Xi$ consistent with the observational law.
 
\begin{proposition}[Structural identification]
\label{prop:structural_identification}
Under \textup{(S1)--(S4)}, with $\pi$ frozen, $\Theta_{\nu}$ is identified if and only if there exists $\xi\in\mathbb R^{N}$ with
\textup{(i)} $\xi\perp\mathcal B$;
\textup{(ii)} $\sum_{i}\xi_{i}x_{i}(A)=c_{\nu}$; and
\textup{(iii)} $\bigl\langle f,\;\sum_{i}\xi_{i}\sigma_{i}(A)
-w_{\nu}\bigr\rangle=0$ for every $f\in\mathcal K$. In that case $\Theta_{\nu}=\sum_{i}\xi_{i}\,\mathbb E(Y_{i}\mid A,C)$, independently of which solution $\xi$ is taken.
\end{proposition}
 
The vector $\xi$ is a contrast on the observed unit means, and the three conditions say that some weighting of units must avoid the nuisance space, reproduce the policy's direct-treatment contrast, and reproduce its exposure profile by distance. If $\mathcal K$ is all bounded measurable functions on $[0,D]$, taking indicators in \textup{(iii)} forces the equality of measures $w_{\nu}=\sum_{i}\xi_{i}\sigma_{i}(A)$, so identified policy signatures lie in a space of dimension at most $N-\dim\mathcal B$ and no policy with $w_{\nu}\notin\operatorname{span}\{\sigma_{i}(A)\}$ is identified, however finely the field is measured.
 
Restricting $\mathcal K$ removes this obstruction, at the cost of quantifiable kernel misspecification. With $\mathcal K$ finite-dimensional, fix a basis $\psi_{1},\ldots,\psi_{J}$ of bounded measurable functions on $[0,D]$, so that $\phi=\sum_{r}\phi_{r}\psi_{r}$, and define
\begin{equation}
    X_{A}
    =
    \bigl[
        x(A),\,
        S[\psi_{1}](A),\ldots,S[\psi_{J}](A)
    \bigr],
    \qquad
    q_{\nu}
    =
    \bigl(
        c_{\nu},\,
        \langle\psi_{1},w_{\nu}\rangle,\,
        \ldots,\,
        \langle\psi_{J},w_{\nu}\rangle
    \bigr)^{\top},
    \label{eq:policy_estimability_vectors}
\end{equation}
Under \textup{(S1)} the estimand factorises across this pair. Writing $\theta_{\mathcal K}=(\tau,\phi_{1},\ldots,\phi_{J})^{\top}$ for the coefficients of $\phi$ in the basis of $\mathcal K$, 
\begin{equation}
    \Theta_{\nu}=q_{\nu}^{\top}\theta_{\mathcal K},
    \label{eq:estimand_factorisation}
\end{equation}
so that $q_{\nu}$ encodes the policy-related content of the estimand and $\theta_{\mathcal K}$ the response. Identification is then the question of whether this pairing is determined by the observational law. Set $R_{A}=P_{\mathcal B^{\perp}}X_{A}$, with $P_{\mathcal B^{\perp}}$ the projection onto the orthogonal complement of $\mathcal B$.  Proposition~\ref{prop:structural_identification} is then the linear-model estimability criterion $q_{\nu}\in\operatorname{row}(R_{A})$, which holds for every $q_{\nu}$ when $R_{A}$ has full column rank, requiring $N\geq\dim\mathcal B+1+J$.The kernel coefficients need not therefore all be identified and a policy effect remains identified whenever it is insensitive to the unidentified directions.
 
A worked instance is the cluster randomised trial with discrete sources, in which $\mathcal B=\operatorname{span}\{\mathbf 1\}$, $x_{i}(A)$ is own-cluster treatment, and every source is eligible for every unit bar $p$ excluded own-cluster sources, so that $\sigma_{i}(A)([0,D])=M_{0}-p\,x_{i}(A)$ for a fixed source count $M_{0}$. Whenever $\mathbf 1\in\mathcal K$, conditions \textup{(i)} and \textup{(ii)} reduce \textup{(iii)} at $f=\mathbf 1$ to $\langle\mathbf 1,w_{\nu}\rangle+p\,c_{\nu}=0$, so no contrast violating this equality is identified at any allocation or sample size. The scalar on the left is the level leverage of \citet{Watson2026}, where the resulting rank shortfall is resolved by a design or assumption on sources of counterfactual information.
 
The criterion is a property of the realised treatment state as much as of the model, since whether $q_{\nu}\in\operatorname{row}(R_{A})$ depends on $A$: a contrast may be identified under one realisation of the placement process and not under another. It is also where flexible spatial and Bayesian response models require care. If $\mathcal A$ is rich, exposure-response directions may be observationally indistinguishable from changes in the latent field, and in the limiting case $\mathcal B=\mathbb R^{N}$ we have $R_{A}=0$ and no non-zero contrast is identified from the conditional mean alone. A proper prior or smoothing penalty will still deliver an estimate or a posterior, but it resolves directions the observational mean structure does not distinguish rather than creating identification from the data. The same holds approximately when $R_{A}$ has full column rank but small singular values, where contrasts loading on those directions are formally identified, weakly supported by the realised exposure geometry, and correspondingly sensitive to regularisation. 

\section{Cameron-Martin Policies}
\label{sec:admissibility}
\subsection{Setting}
\label{sec:field_setting}
We assume that $g_{0}(\cdot\mid C)$ places mass one on states with $a\ll\rho$ and write $\dot A=\dd A/\dd\rho$ for the random exposure surface of the realised state $A$ of \S\ref{sec:estimand}, the capitalised counterpart of the $\dot a$. An exposure surface is determined only $\rho$-almost everywhere. We assume $\rho(\mathcal D)<\infty$ and $\dot A\in L_2(\rho)$ almost surely, and take $\mathbb X=L_2(\rho)$ with its Borel $\sigma$-algebra. Write $P_{0}$ for the law of $\dot A$ under $g_{0}(\cdot\mid C)$. A policy acting through a measurable $\Phi:\mathbb X\to\mathbb X$ induces $\pi=\Phi_{\#}P_{0}$, and policy positivity is the requirement $\Phi_{\#}P_{0}\ll P_{0}$. We fix a jointly measurable version of the field with continuous sample paths, so that the pointwise marginals of $\dot A$ are well defined, and assume $\rho$ and Lebesgue measure on $\mathcal D$ are mutually absolutely continuous with density bounded above and below on compacta, so that the null sets below are the same for both.
 
The results are stated for a field that is Gaussian on a prespecified scale. We assume a \emph{Gaussianising transformation}: a Borel set $\mathbb X_{0}\subseteq\mathbb X$ with $P_{0}(\mathbb X_{0})=1$ and a bimeasurable bijection $\mathcal T:\mathbb X_{0}\to\mathcal T(\mathbb X_{0})$ onto a Borel subset of $\mathbb X$, such that $Z=\mathcal T(\dot A)$ has law $N(\mu,\Sigma)$ with $\Sigma$ trace class. Policies are required to map $\mathbb X_{0}$ into itself. Two cases are of most interest. First, $\mathcal T=\mathrm {id}$, on which policies act additively, and second, $\mathcal T=\log$ with $\mathbb X_{0}=\{f\in\mathbb X:f>0\ \rho\text{-a.e.},\ \log f\in\mathbb X\}$, on which they act multiplicatively. Restriction to $\mathbb X_{0}$ is what makes the logarithmic case well defined, since $\log$ is not a bijection of $L_{2}(\rho)$. For the identity-scale Gaussian examples, we extend the state space to include finite signed measures $a(\dd t)=f(t)\rho(\dd t)$ with $f\in L_2(\rho)$. The preceding definitions and assumptions extend accordingly, with the exposure-moment condition understood as $\int_{\mathscr A}|a|(\mathcal D)\,q(\dd a\mid C)<\infty$ a.s. for every observational or target law $q$ considered, where $|a|$ denotes the total variation measure. The logarithmic examples retain nonnegative exposure states.
 
Two classical facts about Gaussian measures do all the work below \citep[Ch.~2]{Bogachev2007}. The first is the Cameron-Martin theorem. The \emph{Cameron--Martin space} of $N(\mu,\Sigma)$ is $\mathbb H_{\Sigma}=\Sigma^{1/2}(\mathbb X)$, with norm $\|h\|_{\mathbb H_{\Sigma}}=\inf\{\|w\|_{\mathbb X}:\Sigma^{1/2}w=h\}$. Translation by $h$ gives an equivalent law when $h\in\mathbb H_{\Sigma}$ and a mutually singular one otherwise. When $\Sigma$ has infinite rank, $Z-\mu\notin\mathbb H_{\Sigma}$ almost surely, so admissible displacements are smoother than the sample paths themselves. Each $h\in\mathbb H_{\Sigma}$ determines a linear summary of the field, the Paley-Wiener functional $\widehat h$, with
\begin{equation}
    \widehat h(Z-\mu)\sim N\bigl(0,\|h\|^{2}_{\mathbb H_{\Sigma}}\bigr),
    \label{eq:paley_wiener_law}
\end{equation}
so that $\|h\|_{\mathbb H_{\Sigma}}$ is the number of observational standard deviations by which a displacement by $h$ moves the summary it defines. The second is the Feldman--H\'ajek dichotomy: two Gaussian laws on $\mathbb X$ are either equivalent or mutually singular, never absolutely continuous in one direction only. Positivity for a translation policy is therefore a yes-or-no question, and when the answer is no there is an event holding with probability one under the policy and zero under the observational law. All density identities below are understood almost surely under the relevant reference law.
 
\subsection{Which policies admit a weight?}
\label{sec:admissible}
 
\begin{theorem}[Admissible translations]
\label{thm:admissible_class}
Let $Z=\mathcal T(\dot A)\sim N(\mu,\Sigma)$ with $\Sigma$ trace class, let $h\in\mathbb X$ satisfy $\mathcal T(\mathbb X_{0})+h\subseteq\mathcal T(\mathbb X_{0})$, and let $\Phi_{h}=\mathcal T^{-1}\circ(\,\cdot+h)\circ\mathcal T$ with induced policy $\pi_{h}=\Phi_{h\#}P_{0}$. Then $\pi_{h}$ satisfies policy positivity if and only if $h\in\mathbb H_{\Sigma}$, in which case $\pi_{h}\sim P_{0}$ and the policy weight \eqref{eq:policy_rn_weight} is
\begin{equation}
    \omega_{h}(a)
    =\exp\Bigl\{\widehat h\{\mathcal T(\dot a)-\mu\}
      -\tfrac12\|h\|^{2}_{\mathbb H_{\Sigma}}\Bigr\},
    \qquad
    \mathbb E_{g_{0}}(\omega_{h}^{2}\mid C)
    =\exp\bigl\{\|h\|^{2}_{\mathbb H_{\Sigma}}\bigr\}.
    \label{eq:cm_policy_weight}
\end{equation}
Consequently, for every $\chi^2_{\max}\geq0$, $\mathbb E_{g_0}(\omega_h^2 \mid C) \le 1 + \chi^2_{\max}$ if and only if $\|h\|_{\mathbb H_\Sigma} \le \{\log(1+\chi^2_{\max})\}^{1/2}$.
\end{theorem}
 
For a stationary Mat\'ern Gaussian field of smoothness $\nu$ on a bounded Lipschitz domain $\mathcal D\subset\mathbb R^d$, $\mathbb H_\Sigma$ is norm-equivalent to the Sobolev space $H^{\nu+d/2}(\mathcal D)$, whereas its sample paths belong almost surely to $H^s(\mathcal D)$ for every $0\le s<\nu$. Theorem~\ref{thm:admissible_class} therefore says that a policy may displace the exposure surface by a profile smoother than the surface itself, and by nothing else. $\|h\|_{\mathbb H_{\Sigma}}$ measures how far a policy departs from the observed exposure process and by Theorem~\ref{thm:admissible_class} a bound $\chi^{2}_{\max}$ on the relative variance of the weight is exactly a budget $\|h\|_{\mathbb H_{\Sigma}}\leq\{\log(1+\chi^{2}_{\max})\}^{1/2}$. What does not qualify is any displacement carrying the roughness of the field, since sample paths lie outside $\mathbb H_{\Sigma}$ almost surely. 
 
Gaussian-measure equivalence has entered causal inference before, through overlap conditions when the \emph{covariates} are functional or high-dimensional \citep{Ghosh2019}. Here the infinite-dimensional object is the treatment itself, and the comparison is between its observational law and the pushforward of that law under a policy, which makes admissibility a property of the intervention and allows policies to be classified directly. Theorem~\ref{thm:admissible_class} is the functional analogue of the incremental propensity score interventions of \citet{Kennedy2019}, which perturb the treatment law in a direction chosen so that positivity holds by construction. The difference is that the admissible directions here are fixed by the covariance of the field rather than chosen by the analyst.
 
Theorem~\ref{thm:admissible_class} exhibits a class of admissible policies but does not say whether others exist. The next theorem shows that it is the whole class among regular, non-decreasing policies acting pointwise on the Gaussianising scale. Caps, proportional reductions and hinges motivate this pointwise form, although hard caps and non-differentiable hinges require separate arguments because they do not satisfy condition~\textup{(i)} below.
 
\begin{assumption}[Local scaling and quadratic variation]
\label{ass:local_scaling}
The field $Z=\mathcal T(\dot A)$ has continuous sample paths, and there exist $\alpha\in(0,1)$, a unit vector $e$ and a continuous positive $\sigma:\mathcal D\to(0,\infty)$ such that, uniformly on compacta, 
\begin{equation}
    \mathbb E\{Z(t+ve)-Z(t)\}^{2}
    =\sigma^{2}(t)\,|v|^{2\alpha}\{1+o(1)\},
    \qquad v\to0 .
    \label{eq:local_scaling}
\end{equation}
In addition, for every fixed segment $L=\{t_0+ve:v\in[0,\ell]\}\subset\mathcal D$, with $\ell>0$, with, for $c<c'$ in $[0,\ell]$, $\delta_{n}=(c'-c)2^{-n}$, $v_{j}=c+j\delta_{n}$, the quadratic variations 
\begin{equation}
    Q_{n}(Z;[c,c'])
    =\delta_{n}^{1-2\alpha}\sum_{j=0}^{2^{n}-1}
      \bigl\{Z(t_{0}+v_{j+1}e)-Z(t_{0}+v_{j}e)\bigr\}^{2},
    \label{eq:dyadic_qv}
\end{equation}
satisfy 
\begin{equation*}
    Q_n(Z;[c,c'])
    \longrightarrow
    \int_c^{c'}\sigma^2(t_0+ve)\,\dd v
\end{equation*}
almost surely, simultaneously for all $c<c'$ in $\mathscr Q_\ell=(\mathbb Q\cap[0,\ell])\cup\{0,\ell\}$.
\end{assumption}
The local scaling describes short-range second moments, while the quadratic-variation requirement specifies the pathwise signature used below. Both requirements hold for centred stationary Gaussian fields with exponential covariance, taking $\alpha=1/2$, and with Mat\'ern covariance of smoothness $\nu\in(0,1)$, taking $\alpha=\nu$. Their covariance functions satisfy the regularity conditions of the applicable Gaussian quadratic-variation results \citep{Baxter1956,Gladyshev1961}. The theorem below concerns $\alpha\in(0,1)$ and extensions to smoother fields require a separate higher-order formulation and additional regularity conditions.
 
\begin{theorem}[No other pointwise policy]
\label{thm:pointwise_dichotomy}
Let Assumption~\ref{ass:local_scaling} hold, with $\Sigma$ of infinite rank and $\Sigma(t,t)>0$ for $\rho$-almost every $t$. Let the policy act pointwise on the Gaussianising scale, $\Phi=\mathcal T^{-1}\circ\Phi^{\mathcal T}\circ\mathcal T$ with $\Phi^{\mathcal T}(z)(t)=\varphi\{t,z(t)\}$, and suppose $\varphi:\mathcal D\times\mathbb R\to\mathbb R$ satisfies
\begin{enumerate}
\item[(i)] $\varphi$ is jointly continuous, $\partial_u\varphi$ exists and is jointly continuous, and $\varphi(t,\cdot)$ is non-decreasing for $\rho$-almost every $t$;
\item[(ii)] the explicit spatial variation of $\varphi$ is smoother than that of the field, which holds in particular whenever $\varphi(\cdot,u)$ is Lipschitz on $\mathcal D$, uniformly for $u$ in compact sets.
\end{enumerate}
Then $\Phi_{\#}P_{0}\ll P_{0}$ if and only if $\varphi(t,u)=u+h(t)$ for $\rho$-almost every $t$ and every $u\in\mathbb R$, with $h\in\mathbb H_{\Sigma}$; equivalently $\Phi=\mathcal T^{-1}\circ(\,\cdot+h)\circ\mathcal T$. In that case $\Phi_{\#}P_{0}\sim P_{0}$ with weight \eqref{eq:cm_policy_weight}.
\end{theorem}
 
The theorem classifies pointwise policies acting on the Gaussianising scale under (i) and (ii). It does not classify nonlocal policies, policies without a monotone pointwise form, or hard caps and hinges, which fail (i) and are treated in \S\ref{sec:inadmissible}. The condition is that the policy displace the surface without altering its fine-scale structure: $\partial_{u}\varphi\equiv1$, with the remaining displacement smooth enough to lie in $\mathbb H_{\Sigma}$. The proof uses the local variation of Assumption~\ref{ass:local_scaling} as an almost-sure signature of the observational law, statistics of the kind used to recover local properties of Gaussian fields from dense observations \citep{Istas1997,Anderes2009}. Requiring the signature to be preserved under $\Phi_{\#}P_{0}\ll P_{0}$ forces $(\partial_{u}\varphi)^{2}=1$, monotonicity gives $\partial_{u}\varphi=1$, and what remains is a translation, which Theorem~\ref{thm:admissible_class} classifies. The rigidity is of a different kind from the classical sufficient conditions for absolute continuity under nonlinear transformations of Gaussian space \citep{Ramer1974,Kusuoka1984} as within the pointwise class, no transformation other than a translation preserves absolute continuity, whatever its regularity. 

\subsection{How positivity fails}
\label{sec:inadmissible}
A policy fails positivity if it assigns positive probability to an event having observational probability zero. The first mechanism is flattening. A cap $\Phi(f)=\min(f,c)$ sends every concentration above the threshold to $c(t)$ and a hinge $\Phi(f)=(f-\delta)_{+}$ sends every concentration below $\delta(t)$ to zero, so on the set where either binds the transformed surface takes a prescribed value exactly. If the observational marginals are atomless, the event that the surface equals that value on a set of positive $\rho$-measure has observational probability zero and policy probability equal to the probability that the policy binds, and absolute continuity fails whenever that probability is positive. For a continuous field this is whenever the policy does anything at all. Proposition~\ref{prop:flattening} in the Appendix states this for a general pointwise rule with a flat section. What the argument gives is failure of absolute continuity, however a hinge acts as a translation wherever it does not bind, so the argument alone does not establish mutual singularity.
 
The second is rescaling, which fails for a different reason. On the Gaussianising scale $Z\mapsto rZ$ sends $N(\mu,\Sigma)$ to $N(r\mu,r^{2}\Sigma)$. The two laws share a Cameron--Martin space, since $(r^{2}\Sigma)^{1/2}=r\Sigma^{1/2}$, but equivalence also requires the covariance comparison operator to differ from the identity by a Hilbert-Schmidt operator, and here that operator is $(r^{2}-1)\mathrm I$ \citep[Thm.~2.7.2]{Bogachev2007}. Because $\Sigma$ has infinite rank, a non-zero multiple of the identity is not Hilbert--Schmidt, so Feldman-H\'ajek gives mutual singularity for every $r\neq1$. What one realisation pins down here is the local amplitude $\sigma$ of Assumption~\ref{ass:local_scaling}, which the rescaling multiplies by $r$.
 
The third is roughness. A policy acting pointwise and nonlinearly, whether or not it depends on the realised field, displaces $Z$ by $F(Z)$ with $F(z)(t)$ a function of $z(t)$, so the displacement inherits the spatial roughness of $z$ itself and sample paths of a Gaussian field lie outside $\mathbb H_{\Sigma}$ almost surely. Random displacements survive when they are smooth enough in space. If $F:\mathbb X\to\mathbb H_{\Sigma}$ is $\mathbb H_{\Sigma}$- differentiable with $DF(z)$ Hilbert--Schmidt and $\mathrm I+DF(z)$ invertible for every $z$, and $z\mapsto z+F(z)$ is injective, then $\Phi=\mathcal T^{-1}\circ\{\mathrm{id}+F\}\circ\mathcal T$ preserves equivalence, with weight given by the Ramer-Kusuoka formula \citep{Ramer1974,Kusuoka1984}; see \citet{Ustnel2000} for the systematic theory. 
 
\subsection{Positivity and structural modelling}
\label{sec:model_free_displacement}

Failure of policy positivity rules out the weighting identity for that policy. A structural response model can nevertheless support identification under the conditions of Proposition~\ref{prop:structural_identification}. The identity-link model also simplifies the exposure information needed to define the target as the policy enters through its mean displacement,
\begin{equation*}
    m_\Phi(t)
      =\mathbb E_{P_0}\{\Phi(\dot A)(t)-\dot A(t)\}.
\end{equation*}
The signed mean exposure measure then has density $m_\Phi$ with respect to $\rho$. Consequently, two policies with the same mean displacement have the same average causal contrast under the identity-link response model. This follows from linearity and does not require a finite-dimensional kernel basis.

For a pointwise policy, the mean displacement depends only on the exposure distribution at each location. An additive profile $\Phi(\dot a)=\dot a+h$ gives $m_\Phi=h$ without requiring an exposure model. A proportional reduction gives $m_\Phi(t)=(r-1)\mathbb E\{\dot A(t)\}$, so only the mean exposure surface is needed. A cap instead gives $m_\Phi(t) =-\mathbb E\bigl[\{\dot A(t)-c(t)\}_+\bigr]$, which generally requires information about the marginal distribution beyond its mean. Once the target has been specified, identification from the realised conditional means remains governed by Proposition~\ref{prop:structural_identification}.

The above distinguishes the role of the response model from that of Cameron-Martin admissibility. When the exposure itself is Gaussian, an additive profile admits a weight exactly when $h\in\mathbb H_\Sigma$. Such a policy therefore has both a weight and a mean displacement specified directly by the intervention. The known-displacement property also holds for additive profiles outside $\mathbb H_\Sigma$, although their policies fail positivity. For log-Gaussian exposure, a Cameron-Martin translation on the log scale produces $\Phi_h(\dot a)=e^h\dot a$, with $m_{\Phi_h}(t) =\{e^{h(t)}-1\}\mathbb E\{\dot A(t)\}$. This policy may admit a weight while its target on the exposure scale still depends on the mean exposure surface. Positivity alone therefore does not make a structural target independent of the exposure law either.

In practice, there may be multiple exposure mappings that may be plausible when specifying a structural model. Sensitivity to the exposure mapping should be assessed through the effect of the same policy under each candidate model. However, similar goodness of fit does not establish that these effects agree. Indeed, policy positivity, and hence Cameron-Martin admissibility, play a role in the relationship between the accuracy of structural response modeling and policy effect differences. The following result connects agreement of outcome predictions to agreement of policy effects, without imposing a particular response model.

\begin{proposition}[Agreement of policy predictions]
\label{prop:policy_model_agreement}
Work conditionally on $C$, with $P_0$ and $\pi=\Phi_{\#}P_0$ fixed. For $j=1,2$, let $m_j:\mathbb X\to\mathbb R$ be a candidate conditional mean response for $\overline Y$, with $m_j\in L_2(P_0)\cap L_1(\pi)$, and define its model-implied policy contrast by
\begin{equation*}
    \Theta_{\nu}^{(j)}
    =\int_{\mathbb X}m_j(f)\,(\pi-P_0)(\dd f).
\end{equation*}
If $\pi$ is a Cameron-Martin policy with $h\in\mathbb H_\Sigma$, then
\begin{equation}
    \bigl|\Theta_{\nu}^{(1)}-\Theta_{\nu}^{(2)}\bigr|
    \leq
    \left\{\exp\bigl(\|h\|_{\mathbb H_\Sigma}^{2}\bigr)-1\right\}^{1/2}
    \left[
        \mathbb E_{P_0}
        \left\{\bigl(m_1(\dot A)-m_2(\dot A)\bigr)^2\right\}
    \right]^{1/2}.
    \label{eq:policy_model_agreement}
\end{equation}
The multiplicative constant is sharp over the stated response class.

If instead $\pi\not\ll P_0$, there exists $c>0$ such that, for every $\epsilon>0$, two continuous response functions $m_1,m_2:\mathbb X\to[0,1]$ can satisfy
\begin{equation*}
    \|m_1-m_2\|_{L_2(P_0)}<\epsilon,
    \qquad
    \bigl|\Theta_{\nu}^{(1)}-\Theta_{\nu}^{(2)}\bigr|>c.
\end{equation*}
\end{proposition}

The Cameron--Martin norm therefore quantifies the possible amplification of response-model disagreement, even when effects are evaluated through a structural model. A large norm weakens this guarantee and admissibility alone does not ensure small sensitivity. Agreement here is averaged over possible exposure fields under $P_0$, rather than assessed only at the realised field. Within the pointwise class of Theorem~\ref{thm:pointwise_dichotomy}, policies outside the Cameron-Martin class lack this uniform control. A restricted structural model may nevertheless identify their effects and provide stronger guarantees. Sensitivity analyses should therefore compare the effects of the same \textit{policy} across plausible exposure mappings and response assumptions, rather than just prediction accuracy or goodness of fit, alongside checking identification under each model.

\subsection{Reconstructed exposure surfaces}
\label{sec:reconstructed_exposure}
A further issue arises when exposure is reconstructed from incomplete observations. Write $\mathcal O$ for the information used and $\widehat a(t)=\mathbb E\{\dot A(t)\mid \mathcal O\}$ for the reconstruction. An affine $\varphi(t,\cdot)$ commutes with conditional expectation, so additive and proportional policies may be evaluated on
$\widehat a$ exactly. A cap may not: by Jensen's inequality
\begin{equation}
    \{\widehat a(t)-c(t)\}_{+}
      \leq\mathbb E[\{\dot A(t)-c(t)\}_{+}\mid\mathcal O],
    \label{eq:plugin_cap}
\end{equation}
strictly whenever the conditional distribution places mass on both sides of the threshold; if that distribution is Gaussian with standard deviation $s(t)$ and $\widehat a(t)=c(t)$, the gap is $s(t)/\surd(2\pi)$. Applying a cap to a fitted mean surface therefore understates the expected exposure removed, by most where the reconstruction is least certain, and expected removal should instead be computed under the conditional exposure distribution. The direction of the resulting error in the causal contrast depends in addition on the exposure-response relationship, and non-affinity alone does not guarantee a non-zero error. The general statement is Proposition~\ref{prop:plugin_exposure} in the Appendix.

\subsection{Discrete source configurations}
\label{sec:sources}
Discrete source configurations are the contrasting case, and the contrast isolates what is special about fields. Suppose that, conditionally on $C$, $g_{0}$ and $\pi$ are Poisson processes with intensities $\lambda_{0}(\cdot\mid C)$ and $\lambda_{\pi}(\cdot\mid C)$ of finite total mass and $\lambda_{\pi}\ll\lambda_{0}$. Then $\pi\ll g_{0}$ always, and there is no singularity to contend with. The reason is that a configuration is a finite collection of points, so it carries no analogue of the fine-scale structure that Assumption~\ref{ass:local_scaling} makes a single realisation of a field determine with certainty. For $a=\sum_{m=1}^{M}\delta_{T_{m}}$,
\begin{equation}
    \omega(a,C)
    =\exp\Bigl\{\int_{\mathcal D}(\lambda_{0}-\lambda_{\pi})(t\mid C)\,\dd t\Bigr\}
    \prod_{m=1}^{M}\frac{\lambda_{\pi}(T_{m}\mid C)}{\lambda_{0}(T_{m}\mid C)},
    \qquad
    \mathbb E_{g_{0}}(\omega^{2}\mid C)=1+\chi^{2}(\pi\,\|\,g_{0}),
    \label{eq:poisson_weight}
\end{equation}
the product being the familiar inverse-probability form at the realised locations and the exponential compensating for disagreement in the expected number of sources. Placement confounding corresponds to $\lambda_{0}$ depending on components of $C$ on which the target policy does not. Identification and asymptotic theory for stochastic interventions on repeatedly observed spatio-temporal point patterns \citep{Papadogeorgou2022} assume a bounded relative overlap condition between the two intensities and \eqref{eq:poisson_weight} locates that condition as a statement about the second moment of the weight rather than about absolute continuity, which here holds automatically.

Positivity is not, however, sufficient. From \eqref{eq:poisson_weight},
$\chi^{2}(\pi\,\|\,g_{0})=\exp\{\int_{\mathcal D}(\lambda_{\pi}-\lambda_{0})^{2}/\lambda_{0}\,\dd t\}-1$, so if the two intensities differ by a fixed local amount on a fixed fraction of the domain the relative variance of the weight grows exponentially in $|\mathcal D|$. Global weighting therefore degrades as the domain grows, which is the opposite of what an increasing-domain argument requires, and estimation must exploit local structure or the response model. Let $B_{i}\subseteq\mathcal D$ be a prespecified interference neighbourhood for unit $i$, write $\Lambda_{B_{i}}=\int_{B_{i}}\lambda_{0}(t\mid C)\,\dd t$ for the expected local source count, and let $\pi_{i}$ and $g_{0i}$ denote the restrictions of the two laws to $B_{i}$. For a tilted target intensity $\lambda_{\pi}(t\mid C)=\exp\{\alpha+\eta^{\top}\zeta(t,C)\}\lambda_{0}(t\mid C)$, with $\zeta$ a vector of policy-relevant spatial features, the tilted law is again Poisson, so tilting commutes with restriction and the local weight is exact. Writing $\bar\eta=\sup_{t,C}|\alpha+\eta^{\top}\zeta(t,C)|$, the local second moment is at most $\exp\{(e^{\bar\eta}-1)^{2}\Lambda_{B_{i}}\}$, so a budget $\chi^{2}(\pi_{i}\,\|\,g_{0i})\leq\chi^{2}_{\max}$ is enforced
before outcome analysis by
\begin{equation}
    \bar\eta\leq\log\Bigl\{1+\sqrt{\log(1+\chi^{2}_{\max})/\Lambda_{B_{i}}}
    \Bigr\}.
    \label{eq:tilt_weight_bound}
\end{equation}
The field analogue is the Cameron--Martin bound $\|h\|_{\mathbb H_{\Sigma}}\leq\{\log(1+\chi^{2}_{\max})\}^{1/2}$ of \S\ref{sec:admissible}, which by \eqref{eq:paley_wiener_law} asks that the policy displace the linear exposure summary $\widehat h$ by at most that many observational standard deviations. The Cameron--Martin norm is thus the field analogue of the log-intensity perturbation, and the two budgets differ in how they scale. In the source case the admissible perturbation decreases as $\Lambda_{B_{i}}^{-1/2}$, so a larger interference neighbourhood carries more placement information but yields a more variable likelihood ratio. In the field case the corresponding quantity is fixed by the covariance and does not depend on how much of the domain is used. If the target is specified relative to an externally estimated reference intensity rather than $\lambda_{0}$, \eqref{eq:tilt_weight_bound} holds only up to the discrepancy between them. When $\lambda_{0}$ is itself estimated on the study domain, policy construction must likewise be separated from the outcome analysis.

\section{The Role of Discretisation}
\label{sec:discretisation}
A finite representation can admit a policy weight even when the continuum policy law is not absolutely continuous with respect to the observational law. The question is whether the weights remain controlled as increasingly informative representations of the same two laws are examined. The following result characterises square-integrable continuum weights through their projected second moments.

\begin{proposition}[Diagnostics under representation refinement]
\label{prop:diagnostic_futility}
Let $P=P_0$ and $Q=\Phi_\#P_0$, and let $\Pi_m$ be measurable finite-dimensional representations such that
\begin{equation*}    
    \mathcal F_m=\sigma(\Pi_m),
    \qquad
    \mathcal F_m\subseteq\mathcal F_{m+1},
    \qquad
    \sigma\!\left(\bigcup_{m\geq1}\mathcal F_m\right)
      =\mathcal B(\mathbb X).
\end{equation*}
Let $P_m=(\Pi_m)_\#P$ and $Q_m=(\Pi_m)_\#Q$, and suppose
$Q_m\ll P_m$ for every $m$. Write
\begin{equation*}
    \omega^{(m)}
    =
    \frac{\dd Q_m}{\dd P_m},
    \qquad
    \chi_m^2
    =
    \mathbb E_{P_m}\{(\omega^{(m)})^2\}-1,
\end{equation*}
allowing the value $+\infty$. Then
\begin{enumerate}
    \item[(i)] $\chi_m^2$ is non-decreasing in $m$;
    \item[(ii)] if $Q\ll P$ with $\dd Q/\dd P\in L_2(P)$, then $\chi_m^2 \longrightarrow \chi^2(Q\,\|\,P) <\infty$;
    \item[(iii)] if $Q\not\ll P$, or if $Q\ll P$ but     $\dd Q/\dd P\notin L_2(P)$, then $\chi_m^2\longrightarrow+\infty$.
\end{enumerate}
Consequently, existence or finite variance of a policy weight at any fixed representation dimension does not imply continuum policy positivity. What distinguishes a policy admitting a square-integrable continuum weight is boundedness of $\chi_m^2$ under refinement.
\end{proposition}

The same change-of-measure argument gives $\sqrt{\chi_m^2}$ as the sharp factor controlling policy disagreement from observational response disagreement at representation $m$, whenever $\chi_m^2$ is finite. Its divergence therefore also describes the loss of uniform control in structural-model sensitivity analyses. For Gaussian shifts and centred rescalings, the projected second moments have explicit forms.

Two cases make this concrete. Let $Z\sim N(\mu,\Sigma)$ and let $\Pi_{p}$ be representations whose sigma-fields are nested and generate $\mathcal B(\mathbb X)$, with $Z_{p}=\Pi_{p}(Z)\sim N_{p}(\mu_{p},K_{p})$ and $K_{p}$ non-singular. For a fixed shift $h$ represented coherently, so that $\Pi_{p}(Z+h)=Z_{p}+h_{p}$ almost surely with $h_{p}$ deterministic, the representation-level weight satisfies $\mathbb E_{g_{0}}(\omega_{p}^{2})=\exp(h_{p}^{\top}K_{p}^{-1}h_{p})$, increasing to $\exp(\|h\|_{\mathbb H_{\Sigma}}^{2})$ when $h\in\mathbb H_{\Sigma}$ and to $+\infty$ otherwise. For a centred rescaling $Z\mapsto rZ$ with $r>0$, $\mu_{p}=0$ and $\Pi_{p}(rZ)=rZ_{p}$ almost surely,
\begin{equation}
\label{eq:rescaling_weight}
    \mathbb E_{g_{0}}(\omega_{p}^{2})
      =(2r^{2}-r^{4})^{-p/2}\ \ (0<r<\surd2),
    \qquad
    +\infty\ \ (r\geq\surd2),
\end{equation}
so for fixed $r\neq1$ the second moment grows geometrically in $p$. Cameron-Martin shifts therefore have uniformly bounded second moments under refinement, whereas a fixed non-trivial centred rescaling does not, despite admitting a density ratio at every finite dimension. A non-zero mean changes the growth rate in the second case.

Cameron--Martin shifts therefore have uniformly bounded second moments under refinement, whereas a fixed non-trivial centred rescaling has geometrically increasing second moments whenever those moments are finite. A diagnostic at fixed resolution describes the projected problem and provides a lower bound on the continuum divergence, but cannot by itself establish boundedness under further refinement. In general, divergence excludes a square-integrable continuum weight; it does not necessarily exclude absolute continuity.

This gives a concrete interpretation of Condition C4 of \citet{Jiang2026}, which bounds the mean and variance of the density ratio using constants independent of the number of retained functional principal-component scores. Proposition~\ref{prop:diagnostic_futility} identifies the continuum requirement represented by such uniform second-moment control. In particular, the centred Gaussian rescaling above violates this control despite admitting a density ratio at every finite dimension.

These results concern increasingly informative representations of fixed observational and policy laws. Increasing the number of basis functions while changing the fitted covariance or the policy does not automatically satisfy this requirement. For a nonlinear policy, projecting the transformed field can also differ from transforming a truncated field and then projecting it back into the basis. For practice, the scientific policy can remain a transformation of the surface, but its implementation in the chosen representation must be explicit. A surface transformation may leave a finite basis span. Its pushforward remains defined on the ambient field space, but fails absolute continuity if it assigns positive probability outside the observational model's span. Projecting the result back defines a modified policy and should be stated as such. Agreement at grid points need not imply agreement between interpolated surfaces. The representation is therefore an important part of analysis reporting and a refinement study should distinguish changes in resolution from changes in the policy or observational model.

\section{Example: PM2.5 Pollution}
\label{sec:example}
\citet{DefraTargets2023} set two legally binding targets for 2040 for PM2.5 pollution. Both statutory targets of \S1 concern the same modelled surface. Under the classification of \S\ref{sec:admissibility} they lie on opposite sides of  Theorem~\ref{thm:pointwise_dichotomy}. The concentration target is a cap, which flattens the surface on a set of positive
measure and therefore fails policy positivity by the flattening argument of \S\ref{sec:inadmissible}. The exposure reduction target is a proportional reduction, which on the scale that renders the field Gaussian is the translation $h=(\log r)\mathbf 1$ and, the constants lying in $\mathbb H_{\Sigma}$ for the covariances in routine use, is admissible. Compliance with both is assessed from monitoring sites. Here we consider the evaluation of the effects of achieving them, for which the modeled concentration surface is used.

We take the 2021-based background maps of modelled annual mean $\mathrm{PM}_{2.5}$ published for local air quality management, a $1\,$km grid covering the Midlands, UK. The surface is treated as $\widehat a(t)=\mathbb E\{\dot A(t)\mid\mathcal O\}$ in the sense of \S~\ref{sec:reconstructed_exposure}, and a Mat\'ern covariance is fitted to $\log\widehat a$ by variogram, the smooth component alone providing $\Sigma$. Including the nugget would give $\mathbb H_{K}=\mathbb R^{m}$ and render every direction admissible in the representation. Six policies are considered: uniform reductions of $5$, $10$ and $35\%$; removal of the road-transport contribution, obtained from the source attribution accompanying the maps; and caps at $10$ and $5\,\mu\mathrm{g\,m^{-3}}$, the second being the World Health Organization guideline value. No outcome data are used anywhere in this section.

\begin{figure}
    \centering
    \includegraphics[width=\linewidth]{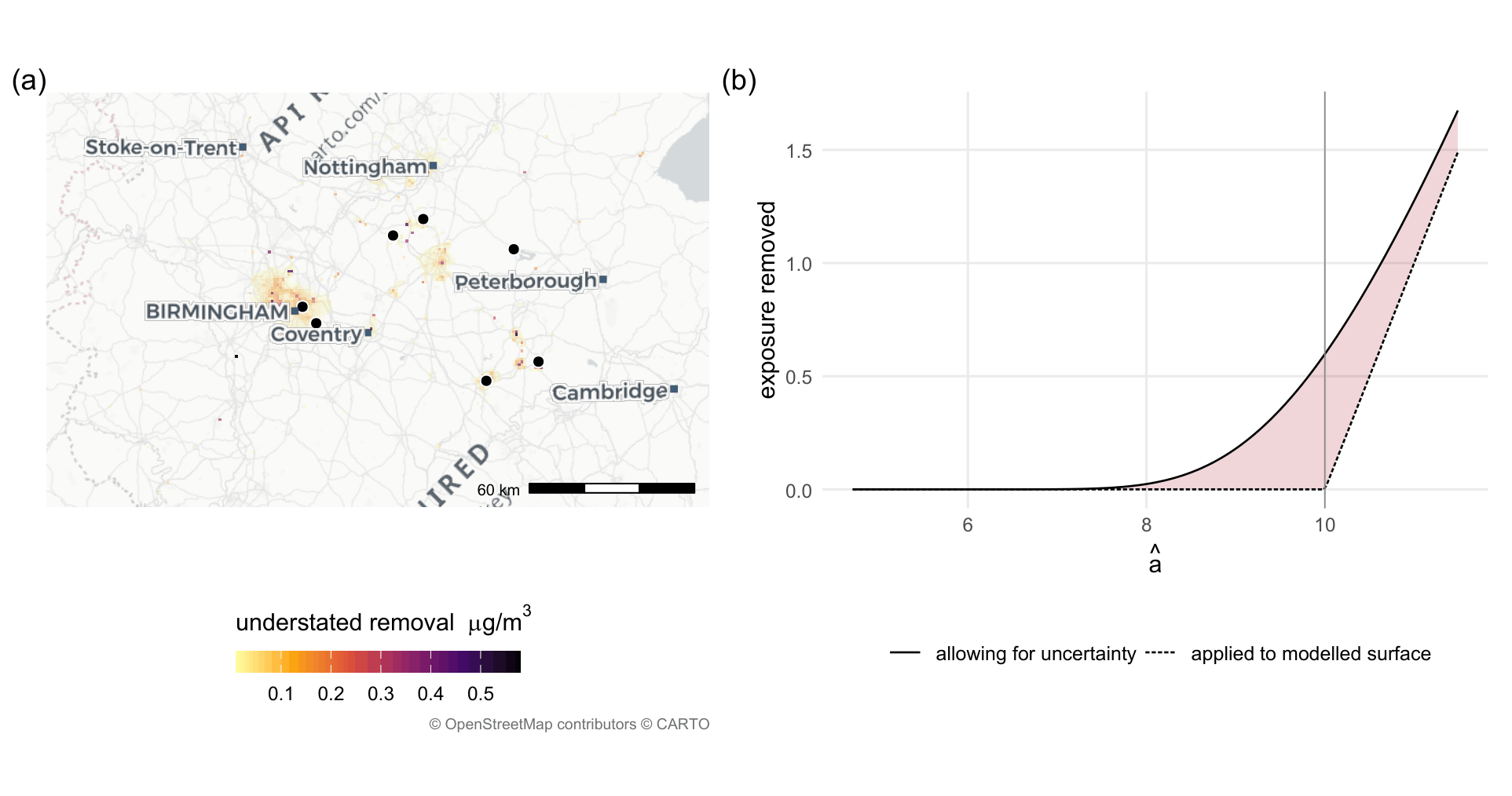}
    \caption{Field policies applied to modelled annual mean $\mathrm{PM}_{2.5}$, Midlands, 2021. (a) Exposure removed by a $10\,\mu\mathrm{g\,m^{-3}}$ cap, $\mathbb E[(\dot A-c)_{+}\mid\mathcal O]-(\widehat a-c)_{+}$ the marked cells are the only ones at which the same cap applied to the modelled surface removes any exposure at all. (b) The two quantities compared in (a) as functions of $\widehat a$, the shaded gap is plugin exposure disrepancy. Basemap \textcopyright\ OpenStreetMap contributors.}
    \label{fig:laqm}
\end{figure}

Figures~\ref{fig:laqm} and \ref{fig:laqm2} report four consequences. The cap applied to the reconstructed surface removes exposure at seven cells, whereas $\mathbb E[(\dot A-c)_{+}\mid\mathcal O]$ is positive across the conurbation and the plug-in understates the exposure removed by 6\%, the discrepancy being largest along the threshold contour where $\widehat a\approx c$ and the reconstruction is least certain. The weight budget $h^{\top}K^{-1}h$ saturates under refinement for the uniform reductions and grows for the caps and for road-sector removal, as Proposition~\ref{prop:diagnostic_futility} requires. Only reductions below about $7\%$ admit a weight with $\chi^{2}\leq100$. A weightable policy may displace the linear exposure summary $\widehat h$ by at most $2.15$ observational standard deviations, and the statutory $35\%$ target displaces it by considerably more. The caps lie outside the weightable class altogether. Finally, panel (d) in Figure~\ref{fig:laqm2} reports the sharp factor of Proposition~\ref{prop:policy_model_agreement}: two candidate response models agreeing to within $\epsilon$ in $L_{2}(P_{0})$ can imply policy contrasts differing by $A\epsilon$. The factor is a ratio of outcome-scale quantities, so it is computed here without an outcome model or an exposure-response mapping; it is the weight budget of panel (c) in the units in which that budget acts on a structural analysis, and it is what remains when no weighting is performed. The $5\%$ reduction gives $A=1.7$, so observational agreement transfers to the policy contrast essentially intact. The $10\%$ reduction gives $A=18$, and the statutory $35\%$ target $A=1.3\times10^{21}$: two response models would have to agree to within $10^{-23}$ of an outcome standard deviation for its effect to be pinned to within a hundredth of one. For the caps and for road-sector removal the factor lies far above this range already at the source resolution and increases without bound under refinement, so in the continuum no degree of observational agreement constrains the policy contrast at all.

\begin{figure}
    \centering
    \includegraphics[width=\linewidth]{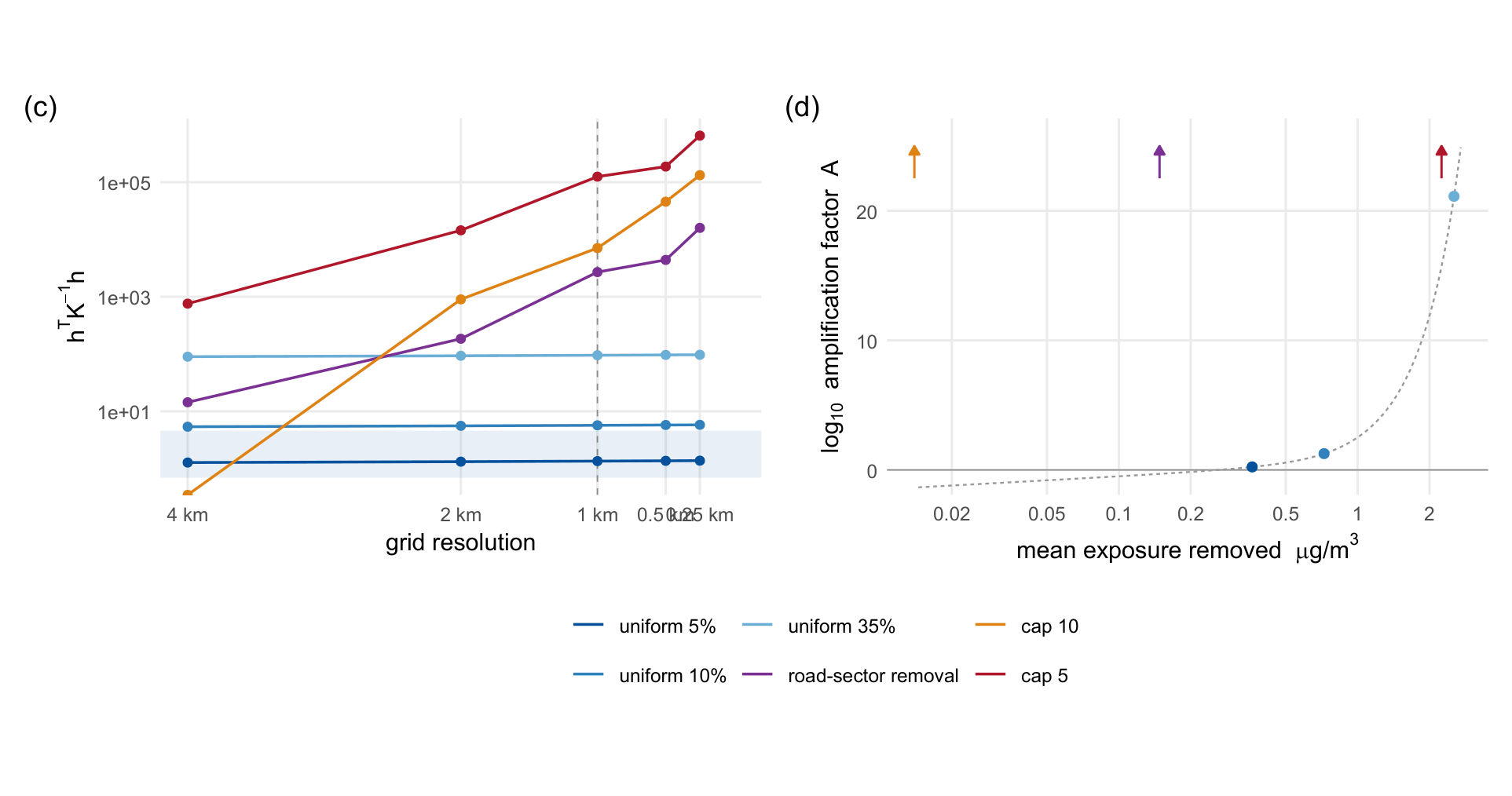}
    \caption{Field policies applied to modelled annual mean $\mathrm{PM}_{2.5}$, Midlands, 2021. (c) Weight budget $h^{\top}K^{-1}h$ against grid resolution; the dashed line marks the resolution of the source data, finer resolutions being implied by the fitted covariance rather than observed, and the horizontal band the region admitting a weight with $\chi^{2}$ between $1$ and $100$. (d) Sharp amplification factor $A$ of Proposition~\ref{prop:policy_model_agreement} against the mean exposure removed by each policy: two response models agreeing to within $\epsilon$ in $L_{2}(P_{0})$ can imply policy contrasts differing by $A\epsilon$. Arrows leaving the frame mark policies whose factor lies far above the plotted range at both resolutions and increases without bound under refinement. Basemap \textcopyright\ OpenStreetMap contributors.}
    \label{fig:laqm2}
\end{figure}

The amount of exposure a policy removes does not govern this. The cap at $5\,\mu\mathrm{g\,m^{-3}}$ and the $35\%$ reduction remove almost the same mean exposure, $2.25$ and $2.53\,\mu\mathrm{g\,m^{-3}}$, and fall on opposite sides of the classification, while the cap at $10\,\mu\mathrm{g\,m^{-3}}$ removes $0.014\,\mu\mathrm{g\,m^{-3}}$, less than any other policy considered, and is the most strongly penalised of the six: it displaces seven isolated cells, and it is exactly the non-smoothness of that displacement, not its size, that the Cameron--Martin norm measures. These factors are worst cases over an unrestricted response class, sharp there; a restricted structural model may do better, and a large factor is not evidence that a particular estimator is biased.

\section{Discussion}
 he consequence for practice may be uncomfortable. A policy that caps or flattens an exposure surface has no weighting representation, at any sample size or measurement resolution. Three responses are available and the choice belongs to the protocol rather than to the analysis. The scientific question may be respecified as a displacement policy, whose mean displacement is then fixed by the intervention itself and which is admissible whenever its profile lies in $\mathbb H_{\Sigma}$. The policy may be retained and identified through the structural response model, at the cost of restricting the kernel class, or it may be replaced by a smoothed approximation, in which case the smoothing parameter and the representation dimension become part of the estimand and must be reported with it. Admissibility is necessary for the first route but not sufficient, and the statutory $35\%$ reduction of \S\ref{sec:example} illustrates the gap as it is a displacement policy on the scale that renders the field Gaussian and admits a weight, yet on the fitted covariance its amplification factor is of order $10^{21}$. In practice many analyses take the second route implicitly, and identification then rests on structural restrictions or regularisation rather than on data.
 
Failure of positivity and weak structural identification are independent conditions. A policy admitting no weight may be exactly identified with a well-conditioned $R_{A}$, and an admissible policy with large Cameron-Martin norm admits a weight of little practical use. They meet, however, in the definition of the estimand, where two scales are in play. Among regular non-decreasing pointwise policies satisfying Assumption~\ref{ass:local_scaling}, Theorem~\ref{thm:pointwise_dichotomy} shows that weightability holds exactly for translations on the scale that renders the field Gaussian. Whether the target can be written down without an exposure model is a separate question, settled on the scale on which the response model is linear: by \S\ref{sec:model_free_displacement} an additive profile on that scale has mean displacement $h$, whereas a cap requires the marginal distribution of the field at each location. The two scales coincide for a Gaussian field under a response model linear in concentration and come apart for a log-Gaussian one, where a Cameron--Martin translation is admissible but its target on the concentration scale still involves the mean exposure surface. Admissibility and model-free specification must therefore be checked separately.
 
The consequences for structural analysis are then of three kinds. The estimand ceases to be a functional of the observational law alone and is defined only relative to $(\mathcal K,\mathcal B)$, so varying the specification varies the target rather than probing the robustness of an estimate. Because it is no longer such a functional, the semiparametric devices that ordinarily guard against outcome-model misspecification are unavailable, as is the comparison against a weighting estimate that would ordinarily detect it. And when the mean displacement itself depends on the exposure law, uncertainty about that law enters the definition of the target and not only its precision. For a cap the relevant quantities are the exceedance probability and the exposure density at the threshold, so the target is most exposure-model-dependent exactly where the policy binds. The same feature makes $\Phi(\widehat a)-\widehat a$ understate $m_{\Phi}$ by most where the reconstruction is least certain, as in Figure~\ref{fig:laqm}(a) and (b). A policy specified directly as a displacement profile removes this dependence entirely, and the choice between the two belongs to the definition of the estimand rather than to the analysis.
 
A further consequence concerns how the analysis is carried out rather than how the estimand is defined. Every computation takes place in a finite representation, and by Proposition~\ref{prop:diagnostic_futility} a representation may admit a policy weight of finite variance when no square-integrable continuum weight exists. A diagnostic computed at working resolution therefore describes the projected problem and bounds the continuum one from below, and the two behaviours are visible on the same surface in Figure~\ref{fig:laqm}(c): the budget saturates under refinement for the uniform reductions and grows without bound for the caps and for road-sector removal. The representation, and any policy-smoothing parameter, should be reported with the estimand, and a refinement study should distinguish changes of resolution from changes of the policy or of the fitted covariance.
 
These distinctions bear on existing environmental-health analyses. \citet{Vicedo-Cabrera2021} estimated mortality attributable to anthropogenic warming by applying estimated temperature--mortality relationships to factual and counterfactual histories, and \citet{Wu2025} evaluated urban-greening scenarios through their predicted effects on heat-related mortality. Flexible Bayesian models have been used to estimate spatially and temporally varying pollution--mortality relationships while simultaneously modelling latent confounding structure \citep[e.g.][]{Zaccardi2026}. These target deterministic counterfactual histories rather than policies of the pushforward form considered here, so our results bear on them only through the identification of the response functional subsequently evaluated, and that is where the risk lies. In models rich enough for spatial effects and exposure-response surfaces to compete for the same variation, the restrictions or priors that resolve the competition are unobjectionable in themselves, but their role changes once the fitted model is used to evaluate an intervention: they may become the assumptions that identify the requested contrast. Proposition~\ref{prop:policy_model_agreement} is what makes the concern quantitative rather than rhetorical, since agreement of fitted response surfaces under the observational exposure law constrains agreement of policy effects only through a factor the policy itself determines. Identification should therefore be assessed for the particular contrast being interpreted causally, separately from posterior convergence, model fit or precision.
 
Several limitations point to further work. The positivity results assume a field Gaussian on a prespecified scale, and the pointwise rigidity theorem assumes in addition local scaling of index $\alpha\in(0,1)$ together with a pathwise quadratic-variation signature; smoother fields require a higher-order formulation, and the surface fitted in \S\ref{sec:example} lies outside that range, so the classification there rests on Theorem~\ref{thm:admissible_class} and on the flattening argument alone. The structural results assume a single realised treatment state, an identity link and a linear kernel class, whereas many applications observe repeated spatio-temporal fields. The response model is also linear in the exposure state: under a nonlinear dose--response, such as the threshold forms common in environmental epidemiology, the mean displacement no longer determines the policy contrast and even an additive profile acquires a target that depends on the exposure law. Extending the theory to these settings, and from exact to weak identification, would permit diagnostics quantifying how strongly a particular contrast is supported by the observed exposure geometry. Developing such diagnostics, together with empirical audits of existing functional and spatial exposure analyses, is the natural next step towards practical guidance on when policy effects are learned from the data and when they depend primarily on structural or regularising assumptions.

\section*{Acknowledgements}
In keeping with the Leiden Declaration on Artificial Intelligence and Mathematics around disclosure of AI tool use, we acknowledge the role of large language models in the development of this work. In particular, the link between policy positivity and Cameron-Martin spaces was first identified by Anthropic's \textit{Claude} LLMs during other related work. Claude was used to support the development of other key statements and proof refinement, and in particular to develop the approach used in the proof of Theorem 2.

\bibliographystyle{plainnat}
\bibliography{adapt}

\appendix
\section{Additional Propositions}

\begin{proposition}[Flattening]
\label{prop:flattening}
Let $\Phi$ act pointwise, $\Phi(f)(t)=\varphi\{t,f(t)\}$, and suppose there exist a Borel set $E\subseteq\mathcal D$ with $\rho(E)>0$, a measurable $w:E\to\mathbb R$ and measurable sets $\mathcal J_{t}\subseteq\mathbb R$ such that $\varphi(t,u)=w(t)$ for all $u\in\mathcal J_{t}$ and $t\in E$. Write $E_{f}=\{t\in E:f(t)\in\mathcal J_{t}\}$ and assume
\begin{enumerate}
    \item[(a)] \emph{the policy binds}: $P_{0}\{\rho(E_{\dot A})>0\}>0$;
    \item[(b)] \emph{no pre-existing plateau}: for $\rho$-almost every
    $t\in E$, the distribution of $\dot A(t)$ under $P_{0}$ is atomless.
\end{enumerate}
Then $\Phi_{\#}P_{0}\not\ll P_{0}$.
\end{proposition}

\begin{proposition}[Policies applied to reconstructed exposure]
\label{prop:plugin_exposure}
Let $\mathcal O$ denote the information used to reconstruct the field, and let $\widehat a(t)=\mathbb E\{\dot A(t)\mid\mathcal O\}$. Assume the conditional expectations below are finite. For $\rho$-almost every $t$:
\begin{enumerate}
\item[(i)]
If $\varphi(t,\cdot)$ is affine, then
\[
    \mathbb E[\varphi\{t,\dot A(t)\}\mid\mathcal O]
      =\varphi\{t,\widehat a(t)\}.
\]

\item[(ii)]
For a cap at $c(t)$,
\[
    \{\widehat a(t)-c(t)\}_+
      \leq
    \mathbb E[\{\dot A(t)-c(t)\}_+\mid\mathcal O].
\]
The inequality is strict whenever the conditional distribution assigns positive probability to both sides of the cap. If that distribution is Gaussian with standard deviation $s(t)$ and $\widehat a(t)=c(t)$, the difference is
\[
    \frac{s(t)}{\sqrt{2\pi}}.
\]
\end{enumerate}
\end{proposition}

\section{Proofs}

\begin{proof}[Equation~\ref{eq:policy_weighting_identity}]
Consistency, conditional exchangeability and the specified conditional-mean versions give, almost surely,
\begin{equation*}
    \mathbb E(Y_i\mid A,C)
    =\int F_i(A,C,u)\,Q_i(\dd u\mid C)
    =\mu_i(A,C).
\end{equation*}
The moment assumption justifies conditional expectation, finite summation and the following change of measure:
\begin{align*}
    \mathbb E_{g_0}\{\omega(A,C)\overline Y\mid C\}
    &=
    \frac1N\sum_{i=1}^{N}
    \int_{\mathscr A}
        \omega(a,C)\mu_i(a,C)\,g_0(\dd a\mid C)\\
    &=
    \frac1N\sum_{i=1}^{N}
    \int_{\mathscr A}\mu_i(a,C)\,\pi(\dd a\mid C)\\
    &=\Psi(\pi).
\end{align*}
The same calculation with weight $1$ gives $\mathbb E_{g_0}(\overline Y\mid C)=\Psi(g_0)$. Subtracting yields
\begin{equation*}
\Theta_\nu
    =\mathbb E_{g_0}\!\left[
        \{\omega(A,C)-1\}\overline Y\mid C
    \right].
\end{equation*}
All identities hold almost surely.
\end{proof}

\begin{proof}[Proposition~\ref{prop:structural_identification}]
No topology on $\mathcal K$ is needed, since the relevant kernel has finite codimension and the argument is linear algebra. For $\delta=(\delta_{\gamma},\delta_{\alpha}, \delta_{\tau},\delta_{\phi})\in\Xi$, define the linear map $T:\Xi\to\mathbb R^{N}$ by
\begin{equation*}
    (T\delta)_{i}
    =
    H_{i}^{\top}\delta_{\gamma}
    +\delta_{\alpha i}
    +\delta_{\tau}x_{i}(A)
    +\langle\delta_{\phi},\sigma_{i}(A)\rangle .
\end{equation*}
Thus, for the true parameter $\theta$, $T\theta$ is precisely the vector $\bigl\{\mathbb E(Y_{1}\mid A,C),\ldots, \mathbb E(Y_{N}\mid A,C)\bigr\}^{\top}.$ By \textup{(S4)}, this vector is all that the observational law tells us about $\theta$. Hence two parameter values are observationally equivalent if and only if their difference belongs to $\ker T$.

Under the identity link, the policy contrast is the linear functional
\begin{equation*}
    \Lambda(\delta)
    =
    c_{\nu}\delta_{\tau}
    +\langle\delta_{\phi},w_{\nu}\rangle .
\end{equation*}
There is no contribution from $(\delta_{\gamma},\delta_{\alpha})$ because these terms do not depend on the treatment state and therefore cancel between the two policies. Consequently, $\Theta_{\nu}$ is identified if and only if $\Lambda(\delta)=0$ for every $\delta\in\ker T$: any two parameter values that give the same observed conditional means must also give the same policy contrast. Suppose this condition holds. Then $\Lambda$ depends on $\delta$ only through $T\delta$, so it defines a linear functional on $\operatorname{ran}T\subseteq\mathbb R^{N}$. Since $\operatorname{ran}T$ is a finite-dimensional subspace, this functional can be represented as an inner product with some $\lambda\in\mathbb R^{N}$. Thus
\begin{equation}
    \Lambda(\delta)=\xi^{\top}T\delta
    \qquad\text{for every }\delta\in\Xi .
    \label{eq:structural_proof_representation}
\end{equation}
Conversely, any such representation immediately implies that $\Lambda$ vanishes on $\ker T$.

It remains to determine when \eqref{eq:structural_proof_representation} holds. Varying $\delta_{\gamma}$ and $\delta_{\alpha}$ while setting the other components to zero gives $\xi\perp\operatorname{col}(H)$ for $\xi\perp\mathcal A$ and hence $\xi\perp\mathcal B$. Varying only $\delta_{\tau}$ gives $\sum_{i=1}^{N}\xi_{i}x_{i}(A)=c_{\nu}$ while varying only $\delta_{\phi}\in\mathcal K$ gives
\begin{equation*}
    \left\langle f,\,
        \sum_{i=1}^{N}\xi_{i}\sigma_{i}(A)-w_{\nu}
    \right\rangle=0
    \qquad\text{for every }f\in\mathcal K.
\end{equation*}
These are exactly conditions \textup{(i)--(iii)}. Finally, taking $\delta=\theta$ in
\eqref{eq:structural_proof_representation} gives
\begin{equation*}
    \Theta_{\nu}
    =\xi^{\top}T\theta
    =\sum_{i=1}^{N}\xi_{i}\,
      \mathbb E(Y_{i}\mid A,C).
\end{equation*}
If $\xi$ and $\xi'$ both satisfy \textup{(i)--(iii)}, then $(\xi-\xi')^{\top}T=0$, so both give the same value.
\end{proof}

\begin{proof}[Theorem~\ref{thm:admissible_class}]
Work conditionally on $C$, and let $\mathcal T:\mathbb X_0\to\mathcal T(\mathbb X_0)$ be the
Gaussianising transformation of Section~\ref{sec:field_setting}, with $Z=\mathcal T(\dot A)\sim N(\mu,\Sigma)$ . By construction,
\begin{equation*}
    \mathcal T\{\Phi_h(\dot A)\}
    =\mathcal T(\dot A)+h
    =Z+h.
\end{equation*}
Thus, on the Gaussianising scale, the policy $\Phi_h$ is simply translation by $h$.

The Cameron-Martin theorem states that the laws of $Z+h$ and $Z$ are equivalent if and only if $h\in\mathbb H_{\Sigma}$, and are mutually singular otherwise \citep[Thm.~2.4.5]{Bogachev2007}. When $h\in\mathbb H_{\Sigma}$, their Radon-Nikodym derivative is 
\begin{equation}
    \frac{\dd\mathcal L(Z+h)}{\dd\mathcal L(Z)}(z)
    =
    \exp\left\{
        \widehat h(z-\mu)
        -\frac12\|h\|_{\mathbb H_{\Sigma}}^{2}
    \right\}.
    \label{eq:cm_proof_density}
\end{equation}

Because $\mathcal T$ is a bimeasurable bijection on the full-measure state space $\mathbb X_0$, equivalence and mutual singularity are preserved when the two laws are transported back through $\mathcal T^{-1}$. Therefore $\pi_h=\Phi_{h\#}P_0$ is equivalent to $P_0$ when $h\in\mathbb H_{\Sigma}$ and mutually singular with $P_0$ otherwise. In particular, $\pi_h\ll P_0$ if and only if $h\in\mathbb H_{\Sigma}$; there is no intermediate case of one-way absolute continuity.

Substituting $z=\mathcal T(\dot a)$ into \eqref{eq:cm_proof_density} gives the policy weight
\begin{equation*}
    \omega_h(a)
    =
    \exp\left\{
        \widehat h\{\mathcal T(\dot a)-\mu\}
        -\frac12\|h\|_{\mathbb H_{\Sigma}}^{2}
    \right\},
\end{equation*}
which is \eqref{eq:cm_policy_weight}.

Finally, under $P_0$ the Paley--Wiener functional $\widehat h(Z-\mu)$ is centred Gaussian with variance $\|h\|_{\mathbb H_{\Sigma}}^{2}$. Hence the Gaussian moment-generating function gives
\begin{align*}
    \mathbb E_{g_0}(\omega_h^2\mid C)
    &=
    \exp\{-\|h\|_{\mathbb H_{\Sigma}}^{2}\}
    \mathbb E\left[
        \exp\{2\widehat h(Z-\mu)\}
        \,\middle|\,C
    \right]\\
    &=
    \exp\{-\|h\|_{\mathbb H_{\Sigma}}^{2}\}
    \exp\{2\|h\|_{\mathbb H_{\Sigma}}^{2}\}\\
    &=
    \exp\{\|h\|_{\mathbb H_{\Sigma}}^{2}\}.
\end{align*}
Therefore, for every $M\geq1$,
\begin{equation*}
    \mathbb E_{g_0}(\omega_h^2\mid C)\leq M
    \quad\Longleftrightarrow\quad
    \|h\|_{\mathbb H_\Sigma}\leq\sqrt{\log M}.
\end{equation*}
Taking $M=1+\chi^2_{\max}$ gives the stated bound.
\end{proof}

\begin{proof}[Theorem~\ref{thm:pointwise_dichotomy}]
Write
\begin{equation*}
    v_{j,n}=j2^{-n}\ell,\qquad
    t_{j,n}=t_0+v_{j,n}e,\qquad
    b_n=(\ell2^{-n})^{1-2\alpha},
\end{equation*}
and, for $0\le c<c'\le\ell$, set
\begin{equation*}
    \mathcal J_n(c,c')
    =
    \{j\in\{1,\ldots,2^n\}:c<v_{j,n}\le c'\}.
\end{equation*}
We first record the two displays referred to in the statement. Fix a segment $L=\{t_{0}+ve:v\in[0,\ell]\}\subset\mathcal D$ in the direction $e$ of Assumption~\ref{ass:local_scaling}, write $v_{j,n}=j2^{-n}\ell$, and for $[c,c']\subseteq[0,\ell]$ define the dyadic quadratic variation
\begin{equation}
    Q_n(f;[c,c'])
    =
    b_n\sum_{j\in\mathcal J_n(c,c')}
    \{f(t_{j,n})-f(t_{j-1,n})\}^{2}.
\end{equation}
By the quadratic-variation clause of Assumption~\ref{ass:local_scaling}, these variations converge almost surely to $\int_c^{c'}\sigma^2(t_0+ve)\,\dd v$ simultaneously for endpoints in $\mathscr Q_\ell$. The same conclusion then holds for every subinterval, since the limiting measure has no atoms. Condition~(ii) of the theorem is that, for every compact $K\subset\mathbb R$ and every such $L$ and $[c,c']$,
\begin{equation}
    b_n\sum_{j\in\mathcal J_n(c,c')}
    \sup_{u\in K}
    \bigl[
        \varphi(t_{j,n},u)-\varphi(t_{j-1,n},u)
    \bigr]^2
    \longrightarrow0.
    \label{eq:phi_spatial_variation}
\end{equation}
If $\varphi(\cdot,u)$ is Lipschitz with constant $L_{K}$ uniformly for $u\in K$, each summand is at most $L_{K}^{2}2^{-2n}\ell^{2}$ and there are at most $2^{n}(c'-c)/\ell+1$ of them, so the left-hand side of \eqref{eq:phi_spatial_variation} is $O(2^{n(2\alpha-2)})$ and vanishes because $\alpha<1$.

Let $P_Z=\mathcal T_{\#}P_0=N(\mu,\Sigma)$ and write $\Psi(z)(t)=\psi\{t,z(t)\}$, so that the policy on the original scale is $\Phi=\mathcal T^{-1}\circ\Psi\circ\mathcal T$. Since $\mathcal T$ is bimeasurable,
\begin{equation*}
    \Phi_{\#}P_0\ll P_0
    \quad\Longleftrightarrow\quad
    \Psi_{\#}P_Z\ll P_Z.
\end{equation*}
We therefore work throughout on the Gaussianising scale. If $\psi(t,u)=u+h(t)$ with $h\in\mathbb H_{\Sigma}$, sufficiency follows immediately from Theorem~\ref{thm:admissible_class}. For necessity, suppose $\Psi_{\#}P_Z\ll P_Z$, and write $W=\Psi(Z)$. Let $\mathscr C_\rho$ be the space of continuous, square-integrable functions on $\mathcal D$, equipped with uniform convergence on compact subsets together with $L_2(\rho)$ convergence. This is a Polish space. Its natural injection $\iota:\mathscr C_\rho\to\mathbb X$ is continuous and injective, since $\rho$ has full support. By the Lusin--Souslin theorem, $\mathbb X_c=\iota(\mathscr C_\rho)$ is Borel in $\mathbb X$ and $\iota^{-1}$ is Borel. Both $Z$ and $W$ belong to $\mathbb X_c$ almost surely, and their pointwise versions are their unique continuous representatives. All point evaluations below use these representatives. In particular, the event $E_L$ below, defined through countably many point evaluations and limits, is Borel in $\mathbb X_c$ and hence in $\mathbb X$.

Fix a segment $L=\{t_0+ve:v\in[0,\ell]\}\subset\mathcal D$ and a subinterval $[c,c']\subseteq[0,\ell]$. At resolution $n$ the interval $[0,\ell]$ is divided into dyadic pieces with grid points $v_{j,n} = j2^{-n}\ell$ for $j = 0,1,\dots,2^n$. Write $t(v) = t_0 + ve$ $z(v) = Z(t(v))$, and $w(v) = \psi(t(v), z(v))$ and let $\Delta_j$ denote an increment over $[v_{j-1,n},v_{j,n}]$.

\emph{Step 1: quadratic variation under a pointwise transformation.}
Decompose the transformed increment as
\begin{equation*}
    \Delta_j w
    =
    \psi\{t(v_{j-1,n}),z(v_{j,n})\}
    -
    \psi\{t(v_{j-1,n}),z(v_{j-1,n})\}
    +
    \Delta_j^{(t)},
\end{equation*}
where $t_{j,n}=t_0+v_{j,n}e$ and
\begin{equation*}
    \Delta_j^{(t)}
    =
    \psi\{t(v_{j,n}),z(v_{j,n})\}
    -
    \psi\{t(v_{j-1,n}),z(v_{j,n})\}.
\end{equation*}
By the mean value theorem,
\begin{equation*}
    \Delta_j w
    =
    \partial_u\psi(t(v_{j-1,n}),\xi_{j,n})\Delta_j z+\Delta_j^{(t)},
\end{equation*}
for some $\xi_{j,n}$ between $z(v_{j-1,n})$ and $z(v_{j,n})$. Almost surely, the continuous path $z$ has compact range $K=z([c,c'])$. By condition~\textup{(ii)},
\begin{equation*}
    b_n
    \sum_j\{\Delta_j^{(t)}\}^2
    \longrightarrow0.
\end{equation*}
Moreover, $\partial_u\psi$ is bounded on the relevant compact set, so Assumption~\ref{ass:local_scaling} implies that $b_n \sum_jg_{j,n}^2(\Delta_jz)^2$ is bounded. The cross term is therefore negligible by the Cauchy--Schwarz inequality. Since $z$ is uniformly continuous and $\partial_u\psi$ is uniformly continuous on compact sets,
\begin{equation*}
    \sup_j
    \left|
        \partial_u\psi\{t(v_{j-1,n}),\xi_{j,n}\}
        -
        \partial_u\psi\{t(v_{j-1,n}),z(v_{j-1,n})\}
    \right|
    \longrightarrow0,
\end{equation*}
Using Assumption~\ref{ass:local_scaling} on a finite partition of $[c,c']$ and approximating the continuous function $[\partial_u\psi\{t(v),z(v)\}]^2$ by step functions now gives for the dyadic quadratic variation defined in \eqref{eq:dyadic_qv}
\begin{equation}
    Q_n(W;[c,c'])
    \longrightarrow
    \int_c^{c'}
        \left(\partial_u\psi\{t(v),z(v)\}\right)^2\sigma^2(t(v))\,\dd v
    \qquad P_Z\text{-almost surely}.
    \label{eq:qv_chain_rule}
\end{equation}

\emph{Step 2: preservation of the pathwise signature forces unit derivative.}
For the fixed segment $L$, let $E_L$ be the event that the quadratic variation in Assumption~\ref{ass:local_scaling} has its stated limit on every subinterval with rational endpoints. This is a measurable event and, by countability and Assumption~\ref{ass:local_scaling},
\begin{equation*}
    P_Z(E_L)=1.
\end{equation*}
Since $\Psi_{\#}P_Z\ll P_Z$,
\begin{equation*}
    P_Z(W\in E_L)
    =
    \Psi_{\#}P_Z(E_L)
    =
    1.
\end{equation*}
Hence, comparing the defining limit of $E_L$ with
\eqref{eq:qv_chain_rule}, almost surely
\begin{equation*}
    \int_a^b
        \{\left(\partial_u\psi\{t(v),z(v)\}\right)^2-1\}\sigma^2(t(v))\,\dd v
    =0
\end{equation*}
for every rational $[a,b]\subseteq[0,\ell]$.

The function $\partial_u\psi\{t(v),z(v)\}$ is continuous and $\sigma^2$ is strictly positive. It follows that
\begin{equation*}
    \left(\partial_u\psi\{t(v),z(v)\}\right)^2=1
    \qquad\text{for every }v\in[0,\ell].
\end{equation*}
Indeed, if $\left(\partial_u\psi\{t(v),z(v)\}\right)^2-1$ were non-zero at some $v_0$, continuity would give a subinterval on which it had one strict sign, contradicting the preceding integral identity. Now fix a point $t$ for which $\Sigma(t,t)>0$ and $\psi(t,\cdot)$ is non-decreasing, and choose a segment of the above form containing $t$. The preceding argument gives
\begin{equation*}
    \left(\partial_u\psi(t,Z(t))\right)^2=1
    \qquad P_Z\text{-almost surely}.
\end{equation*}
Because $Z(t)$ is a non-degenerate Gaussian random variable, its density is strictly positive on $\mathbb R$. Therefore
\begin{equation*}
    \left(\partial_u\psi(t,u)\right)^2=1
\end{equation*}
for Lebesgue-almost every $u$, and hence for every $u$ by continuity of $\partial_u\psi(t,\cdot)$. A continuous function taking values only in $\{-1,1\}$ must be constant. Since $\psi(t,\cdot)$ is non-decreasing, its derivative cannot equal $-1$, and consequently
\begin{equation*}
    \partial_u\psi(t,u)=1
    \qquad\text{for every }u.
\end{equation*}
Thus, for $\rho$-almost every $t$,
\begin{equation*}
    \psi(t,u)=u+h(t),
    \qquad
    h(t)=\psi(t,0).
\end{equation*}
Equivalently, $\Psi(z)=z+h$ as an element of $L_2(\rho)$. Since $W-Z=h$ $\rho$-almost everywhere almost surely, and both $W$ and $Z$ belong to $L_2(\rho)$ almost surely, the deterministic displacement $h$ belongs to $L_2(\rho)$.
 
\emph{Step 3: the displacement must be Cameron--Martin.}
The transformed law is now the law of $Z+h$. By the Cameron--Martin theorem, this law is absolutely continuous with respect to the law of $Z$ if and only if
\begin{equation*}
    h\in\mathbb H_{\Sigma},
\end{equation*}
in which case the two laws are equivalent \citep[Thm.~2.4.5]{Bogachev2007}. The policy on the original scale is therefore
\begin{equation*}
    \Phi
    =
    \mathcal T^{-1}\circ(\,\cdot+h)\circ\mathcal T,
\end{equation*}
with weight \eqref{eq:cm_policy_weight}, completing the proof.
\end{proof}

\begin{proof}[Proposition~\ref{prop:flattening}]
For $f\in\mathbb X$, define
\begin{equation*}
    G(f)=\rho\{t\in E:f(t)=w(t)\},
    \qquad
    S=\{f\in\mathbb X:G(f)>0\}.
\end{equation*}
These definitions do not depend on the representative of $f$.
To establish measurability, set
\begin{equation*}
    G_m(f)
    =
    \int_E \exp\{-m|f(t)-w(t)|\}\,\rho(\dd t),
    \qquad m\ge1.
\end{equation*}
Since $\rho(E)<\infty$,
\begin{equation*}
    |G_m(f)-G_m(g)|
    \le
    m\,\rho(E)^{1/2}\|f-g\|_{L_2(\rho)},
\end{equation*}
so each $G_m$ is continuous. Dominated convergence gives $G(f)=\lim_{m\to\infty}G_m(f)$. Thus $G$ and $S$ are measurable.

Under the observational law, assumption~\textup{(b)} implies that, for $\rho$-almost every $t\in E$, $P_0\{\dot A(t)=v(t)\}=0$. Hence Tonelli's theorem gives
\begin{align*}
    \mathbb E_{P_0}\{G(\dot A)\}
    &=
    \int_E
        P_0\{\dot A(t)=v(t)\}\,\rho(\dd t)\\
    &=0.
\end{align*}
Since $G(\dot A)\geq0$, it follows that $G(\dot A)=0$ almost surely, and therefore $P_0(S)=0$.

Now suppose the policy binds for a field $f$, so that $\rho(E_f)>0$. By definition of $E_f$, for every $t\in E_f$ we have $f(t)\in J_t$, and hence the flattening property gives $\Phi(f)(t)=v(t)$. Therefore
\begin{equation*}
    G\{\Phi(f)\}
    \geq
    \rho(E_f)
    >0,
\end{equation*}
so $\Phi(f)\in S$. Thus
\begin{equation*}
    \{f:\rho(E_f)>0\}
    \subseteq
    \Phi^{-1}(S).
\end{equation*}
Using assumption~\textup{(a)},
\begin{align*}
    \Phi_{\#}P_0(S)
    &=
    P_0\{\Phi(\dot A)\in S\}\\
    &\geq
    P_0\{\rho(E_{\dot A})>0\}\\
    &>0.
\end{align*}
Thus $S$ is a $P_0$-null set to which the transformed law assigns
positive probability. Consequently $\Phi_{\#}P_0\not\ll P_0$.
\end{proof}

\begin{proof}[Proposition~\ref{prop:policy_model_agreement}]
Write $e=m_1-m_2$. For a Cameron--Martin policy, the
Radon--Nikodym identity gives
\[
    \Theta_{\nu}^{(1)}-\Theta_{\nu}^{(2)}
    =\mathbb E_{P_0}\{(\omega_h-1)e(\dot A)\}.
\]
By Theorem~\ref{thm:admissible_class},
\[
    \mathbb E_{P_0}\{(\omega_h-1)^2\}
    =\exp\bigl(\|h\|_{\mathbb H_\Sigma}^{2}\bigr)-1.
\]
Cauchy--Schwarz proves \eqref{eq:policy_model_agreement}.
Taking $m_1=\omega_h$ and $m_2=1$ shows sharpness when $h\neq0$;
the case $h=0$ is immediate.

Now suppose $\pi\not\ll P_0$. By regularity of Borel probability
measures on the separable Hilbert space $\mathbb X$, there is a
compact set $K\subseteq\mathbb X$ with $P_0(K)=0$ and $\pi(K)>0$.
Define
\[
    e_n(f)=\bigl\{1-n\,\operatorname{dist}_{\mathbb X}(f,K)\bigr\}_{+}.
\]
Each $e_n$ is continuous, takes values in $[0,1]$, and converges
pointwise to $\mathbf1_K$. Dominated convergence gives
\[
    \mathbb E_{P_0}\{e_n(\dot A)^2\}\longrightarrow0,
    \qquad
    \int_{\mathbb X}e_n(f)\,(\pi-P_0)(\dd f)
    \longrightarrow\pi(K)>0.
\]
Set $c=\pi(K)/2$. For any $\epsilon>0$, taking $m_1=e_n$ and
$m_2=0$ for sufficiently large $n$ proves the second assertion.
\end{proof}

\begin{proof}[Proposition~\ref{prop:diagnostic_futility}]
Let $\mathcal F_m=\sigma(\Pi_m)$. Since the representations are nested, $\mathcal F_m\subseteq\mathcal F_{m+1}$. The Radon-Nikodym derivatives of the restrictions of $Q$ to these $\sigma$-fields form a non-negative $P$-martingale. In particular,
\begin{equation}
    \omega^{(m)}
    =
    \mathbb E_P\{\omega^{(m+1)}\mid\mathcal F_m\}.
    \label{eq:weight_martingale}
\end{equation}
Conditional Jensen's inequality therefore gives
\begin{align*}
    \mathbb E_P\{(\omega^{(m)})^2\}
    &=
    \mathbb E_P\left[
        \left\{
        \mathbb E_P(\omega^{(m+1)}\mid\mathcal F_m)
        \right\}^2
    \right]\\
    &\leq
    \mathbb E_P\{(\omega^{(m+1)})^2\},
\end{align*}
which proves~\textup{(i)}.

Now suppose $Q\ll P$ and write $\omega=\frac{\dd Q}{\dd P}$, then
$ \omega^{(m)} =\mathbb E_P(\omega\mid\mathcal F_m)$. If $\omega\in L_2(P)$, the $L_2$ martingale convergence theorem and the fact that $\mathcal F_m$ increases to the Borel $\sigma$-field give
\begin{equation*}
    \omega^{(m)}
    \longrightarrow
    \omega
    \qquad\text{in }L_2(P).
\end{equation*}
Hence
\begin{equation*}
    \chi_m^2
    \longrightarrow
    \mathbb E_P(\omega^2)-1
    =
    \chi^2(Q\,\|\,P),
\end{equation*}
which proves~\textup{(ii)}.

It remains to consider the complementary cases. Suppose, contrary to the claim, that
\begin{equation*}
    \sup_m
    \mathbb E_P\{(\omega^{(m)})^2\}<\infty.
\end{equation*}
Then $\{\omega^{(m)}\}$ is an $L_2$-bounded martingale and therefore converges in $L_2(P)$ to some $\omega^{(\infty)}\in L_2(P)$. For every event $A\in\mathcal F_m$,
\begin{equation*}
    Q(A)
    =
    \mathbb E_P\{\omega^{(m)}\mathbf 1_A\}
    =
    \mathbb E_P\{\omega^{(\infty)}\mathbf 1_A\}.
\end{equation*}
Since $\bigcup_m\mathcal F_m$ generates the Borel $\sigma$-field, this identity extends to every Borel event. Thus
\begin{equation*}
    Q\ll P,
    \qquad
    \frac{\dd Q}{\dd P}
    =
    \omega^{(\infty)}
    \in L_2(P).
\end{equation*}
Therefore, if either $Q\not\ll P$ or its Radon--Nikodym derivative is not square-integrable, the sequence cannot be bounded. By part~\textup{(i)} it is non-decreasing, and hence
\begin{equation*}
    \chi_m^2\longrightarrow+\infty.
\end{equation*}
This proves~\textup{(iii)}.
\end{proof}

\begin{proof}[Proposition~\ref{prop:plugin_exposure}]
Part~\textup{(i)} follows from linearity of conditional expectation. For part~\textup{(ii)}, conditional Jensen's inequality applies to the convex function $u\mapsto\{u-c(t)\}_+$. If $\widehat a(t)\leq c(t)$, the gap is positive whenever there is positive conditional probability above the cap. If $\widehat a(t)>c(t)$, the gap equals $\mathbb E[\{c(t)-\dot A(t)\}_+\mid\mathcal O]$ and is positive whenever there is positive conditional probability below the cap. At the threshold, the Gaussian gap is $s(t)\mathbb E(G_+)=s(t)/\sqrt{2\pi}$ for $G\sim N(0,1)$.
\end{proof}

\begin{proof}[Equation~\eqref{eq:rescaling_weight}]
For part~\textup{(i)}, the finite-dimensional Gaussian
change-of-measure formula gives
\[
    \mathbb E_{g_0}(\omega_p^2)
      =\exp(h_p^\top K_p^{-1}h_p).
\]
If $h\in\mathbb H_\Sigma$, the continuum translation has a
square-integrable weight with second moment
$\exp(\|h\|_{\mathbb H_\Sigma}^{2})$.
If $h\notin\mathbb H_\Sigma$, its law is singular with respect
to the observational law. The monotonicity and the two limits
therefore follow from
Proposition~\ref{prop:diagnostic_futility}, applied on the
Gaussianising scale.

For part~\textup{(ii)}, let $p_0$ and $p_1$ be the densities of
$N_p(0,K_p)$ and $N_p(0,r^2K_p)$. When $0<r<\sqrt2$,
Gaussian integration gives
\begin{align*}
    \int_{\mathbb R^p}\frac{p_1(z)^2}{p_0(z)}\,\dd z
    &=
    |r^2K_p|^{-1}|K_p|^{1/2}
    \bigl|(2r^{-2}-1)K_p^{-1}\bigr|^{-1/2}\\
    &=(2r^2-r^4)^{-p/2}.
\end{align*}
For $r\geq\sqrt2$, the integral diverges.
Finally,
$2r^2-r^4=1-(r^2-1)^2\in(0,1)$ when
$0<r<\sqrt2$ and $r\neq1$, proving the geometric growth.
\end{proof}

\end{document}